\documentclass[a4paper,11pt]{article} 

\usepackage[english]{babel}
\usepackage[a4paper]{geometry}
\usepackage{enumerate}  
\usepackage{amsmath}
\usepackage{dsfont}
\usepackage{graphicx}
\usepackage{color}
\usepackage{amsthm}
\usepackage{amssymb}
\usepackage{yhmath}

\makeatletter
\let\@fnsymbol\@arabic 
\makeatother     

\usepackage{hyperref}    

\newcommand{\di}{\textrm{d}}
\newcommand{\hc}{\mbox{h.c.}}

\usepackage{ulem}  

\let\a=\alpha     \let\g=\gamma     \let\d=\delta     \let\e=\varepsilon
\let\z=\zeta        \let\k=\kappa     \let\l=\lambda
\let\m=\mu    \let\n=\nu                      \let\r=\rho
\let\s=\sigma \let\t=\tau         \let\ph=\varphi   
   \let\o=\omega     
 \let\D=\Delta       \let\L=\Lambda

\def\cE{{\cal E}}\def\cV{{\cal V}}
\def\cC{{\cal C}}\def\cF{{\cal F}}\def\cH{{\cal H}}
\def\cN{{\cal N}}\def\cZ{{\cal Z}}
\def\cR{{\cal R}}\def\cL{{\cal L}}\def\cQ{{\cal Q}}
\def\cD{{\cal D}}\def\cG{{\cal G}}
\def\cO{{\cal O}}\def\cK{{\cal K}}

  \def\v0{{\vec 0}}

\def\bal{{\bar \l}}

\def\bR{\mathbb{R}}

\def\cV{\mathcal{V}}
\def\cF{\mathcal{F}}
\def\cG{\mathcal{G}}
\def\cL{\mathcal{L}}
\def\cN{\mathcal{N}}
\def\cE{\mathcal{E}}
\def\cK{\mathcal{K}}
\def\cH{\mathcal{H}}
\def\eps{\varepsilon}
\def\ph{\varphi}

\def\NNN{\mathbb{N}}  
\def\ZZZ{\mathbb{Z}}

\def\bN{\mathbb{N}}  
\def\bZ{\mathbb{Z}} 
\def\bR{\mathbb{R}}

\def\indic{\hbox{\raise-2pt \hbox{\indbf 1}}}

\let\==\equiv

\let\0=\noindent

\def\*{{\hfill\break\null\hfill\break}}

\def\bmedia#1{{\bigl\langle#1\bigr\rangle}}

\def\tende#1{\,\vtop{\ialign{##\crcr\rightarrowfill\crcr
			\noalign{\kern-1pt\nointerlineskip}
			\hskip3.pt${\scriptstyle #1}$\hskip3.pt\crcr}}\,}
\def\otto{\,{\kern-1.truept\leftarrow\kern-5.truept\to\kern-1.truept}\,}
\def\fra#1#2{{#1\over#2}}

\def\wt#1{\widetilde{#1}}

\def\sqt[#1]#2{\root #1\of {#2}}

\def\hc{{\rm h.c.}\,}

\def\wt{\widetilde}

\def\be{\begin{equation}}
	\def\ee{\end{equation}}
\def\bea{\begin{eqnarray}}\def\eea{\end{eqnarray}}
\def\bean{\begin{eqnarray*}}\def\eean{\end{eqnarray*}}
\def\bfr{\begin{flushright}}\def\efr{\end{flushright}}
\def\bc{\begin{center}}\def\ec{\end{center}}
\def\bal{\begin{align}} 
	\def\eal{\end{align}}

\def\spl#1\spl{\[ \begin{split}#1\end{split} \]}

\def\bd{\begin{description}}\def\ed{\end{description}}
\def\bv{\begin{verbatim}}\def\ev{\end{verbatim}}

\def\Halmos{\hfill\vrule height10pt width4pt depth2pt \par\hbox to \hsize{}}

\newtheorem{theorem}{Theorem}[section]
\newtheorem{prop}{Proposition}[section]
\newtheorem{lemma}[prop]{Lemma} 
\theoremstyle{remark}

\numberwithin{equation}{section}
\def \aa{{\mathfrak a}}

\title{The excitation spectrum of two dimensional Bose gases \\ in the Gross-Pitaevskii regime}

\author{
Cristina Caraci$^{*,}$\footnote{Electronic mail: cristina.caraci@math.uzh.ch}\;, 
Serena Cenatiempo$^{\circ,}$\footnote{Electronic mail: serena.cenatiempo@gssi.it}\;, and Benjamin Schlein$^{*,}$\footnote{Electronic mail: benjamin.schlein@math.uzh.ch} \\[0.2cm]
{\footnotesize $^{*}$Institute of Mathematics, University of Zurich, Winterthurerstrasse 190, 8057 Zurich.}\\
{\footnotesize $^{\circ}$Gran Sasso Science Institute, Viale Francesco Crispi 7, 67100 L'Aquila, Italy
}}

\date{\today}

\begin{document}
\maketitle

\begin{abstract}
We consider a system of $N$ bosons, in the two-dimensional unit torus. We assume particles to interact through a repulsive two-body potential, with a scattering length that is exponentially small in $N$ (Gross-Pitaevskii regime). In this setting, we establish the validity of the predictions of Bogoliubov theory, determining the ground state energy of the Hamilton operator and its low-energy excitation spectrum, up to errors that vanish in the limit $N \to \infty$. 
\end{abstract}

\section{Introduction}

In the past decades, Bose-Einstein condensates (BEC) have emerged as important quantum systems, in view of the precision and flexibility with which they can be manipulated. Experiments on thin films \cite{Bishop} or in highly elongated magnetic and pancake-shaped optical traps (see e.g. \cite[Sect.\,1.6]{B-PhD}) have also pushed forward the study of BEC in low dimensional systems. As a matter of fact, dimensionality plays a crucial role in situations where spontaneous symmetry breaking of continuous symmetries occurs \cite{MW, H}. Hence, it is not surprising that equilibrium properties of the two dimensional Bose gas exhibit significant differences compared with the three dimensional case 
(see e.g. \cite[Chapter 3]{LSSY}, \cite[Chapter 23]{SP}, \cite{CG}). 

In this paper, we are interested in the low energy spectrum of  two dimensional dilute Bose gases, describing systems where both the quantum and thermal motions are frozen in one direction (see \cite{SY,Bossmann} for a discussion of regimes where the confined system has rather a three dimensional character). 
In particular, we consider $N$ bosons moving in the two-dimensional box $\Lambda = [-1/2 ; 1/2]^{ 2}$, with periodic boundary conditions (the two-dimensional unit torus) and described by the Hamilton operator 
\begin{equation}\label{eq:Ham0} 
H_N = \sum_{j=1}^N -\Delta_{x_j} +  \sum_{i<j}^N e^{2N} V (e^N (x_i -x_j))\,,\end{equation}
acting on the Hilbert space $L^2_s (\Lambda^N)$, the subspace of $L^2 (\Lambda^N)$ consisting of functions that are symmetric with respect to permutations of the $N$ particles. Here $V$ is a non-negative, compactly supported and spherically symmetric two body potential. 
The form of the scaled interaction $V_N(x)=e^{2N} V (e^N x)$ is chosen so that the scattering length of $V_N$ is equal to $e^{-N}\aa$, with $\aa$ the scattering length of $V$. 
Indeed in two dimensions and for a potential $V$ with finite range $R_0$ the scattering length is defined by
\be \label{eq:defa0}
\frac{2\pi}{\log(R/\aa)} =\inf_{\phi} \int_{B_R}  \left[ |\nabla \phi|^2 + \frac 1 2 V |\phi|^2   \right] dx 
\ee
where $R > R_0$, $B_R$ is the disk of radius $R$ centered at the origin and the infimum is taken over functions $\phi \in H^1(B_R)$ with $\phi (x)=1$ for all $x$ with $|x|=R$ (see for example \cite[Sect.\,6.2]{LSSY}). 
In the scaling limit defined by \eqref{eq:Ham0}, known as the two dimensional Gross-Pitaevskii regime, we provide an expression for the ground state energy and the low-energy excitation spectrum,  up to errors vanishing as $N \to \infty$, validating the predictions of Bogoliubov theory \cite{B}. In particular we exhibit a proof of the linear dependence of the dispersion of low-lying excitation at low-momenta, a fact which is interpreted in the physics literature as a signature for superfluidity.

Remark that, rescaling lengths, the two dimensional Gross-Pitaevskii regime can be interpreted as describing an extended Bose gas (of particles interacting through the unscaled potential $V$) at a density that is exponentially small in $N$. While the exponential smallness of the density (or, equivalently, of the scattering length) makes it difficult to directly apply our results to physically relevant situations, it 
should be stressed that the Gross-Pitaevskii regime provides a first example of scaling limit in which peculiarities of two-dimensional systems 
can be observed. 

The following theorem is our main result.  
\begin{theorem} \label{thm:main}
	Let $V \in L^3(\bR^2)$ be non-negative, compactly supported and spherically symmetric. The ground state energy $E_N$ of the Hamiltonian $H_N$ defined in \eqref{eq:Ham0} is such that, as $N \to \infty$,  
\be \begin{split} \label{eq:EN}
E_N=\; & 2 \pi (N-1) + \pi^2 \frak{a}^2 \\
& + \fra 12 \sum_{p\in \L^*_+} \Big[\sqrt{|p|^4+8\pi p^2}-p^2-4\pi +\frac{(4\pi)^2}{2p^2} \big(1- J_0(|p|\frak{a}) \big)\Big]+ \cO(N^{-\frac{1}{10} +\delta})
\end{split}\ee
for any $\delta > 0$. Here, we use the notation $\L^*_+=2\pi \ZZZ^2\setminus\{0\}$ and $J_0$ indicates the zero-th order Bessel function of the first kind (taking into account that $J_0 (r) \sim r^{-1/2}$ as $r \to \infty$ and expanding $\sqrt{|p|^4 + 8\pi p^2}$ for large $|p|$, it is easy to check that the sum on the second line converges). 

Moreover, the spectrum of $H_N - E_N$ below a threshold $\zeta > 0$ consists of eigenvalues having the form 
\begin{equation}\label{eq:exc-main}
\sum_{p \in \L^*_+}	n_p \sqrt{|p|^4+8\pi p^2} + \cO(N^{-\frac{1}{10}+\delta} (1+ \zeta^{17}))
\end{equation} 
for any $\delta > 0$. Here $n_p \in \NNN$ for all $p\in \L^*_+$, and $n_p \neq 0$ for finitely many $p\in \L^*_+$ only. 
\end{theorem}

\medskip

{\it Remarks.} 
\begin{enumerate}[i)]
\item To keep our analysis as simple as possible, we restrict our attention to bosons moving in the two-dimensional unit torus. Our results could be extended to more general trapping potentials, combining the proof of Theorem \ref{thm:main} with ideas from \cite{NNRT,BSS,NT,BSS1}, recently developed in the three dimensional setting. 

\item  To leading order, the first rigorous computation of the ground state energy of a dilute two-dimensional Bose gas has been obtained in \cite{LY2d}. In this paper, the authors considered a system of $N$ particles, moving in a box with side length $L$ and interacting through a two-body potential with scattering length $\frak{a}$. In the thermodynamic limit $N, L \to \infty$ at fixed density $\rho = N / L^3$, they considered the ground state energy per particle, $e (\rho)$, and they proved that 
\[ \Big| e(\rho) - \frac{4\pi \rho}{| \log (\rho \frak{a}^2)|} \Big| \leq  \frac{C \rho}{\left| \log (\rho \frak{a}^2) \right|^{6/5}} \, . \]
Translating to the Gross-Pitaevskii regime (where $\rho = N$ and the scattering length is given by $e^{-N} \frak{a}$), this bound implies that 
\begin{equation}\label{eq:EN-LY} E_N = \frac{4 \pi N^2}{ | \log (N e^{-2N} \frak{a}^2)|} \Big[1 + \cO (N^{-1/5}) \Big] \end{equation} 
which is consistent with the leading order term in (\ref{eq:EN}). The estimate (\ref{eq:EN-LY}) has been extended to general trapping potentials in \cite{LSeY2d}. Recently, also the free energy of a two-dimensional dilute Bose gas at positive temperature has been computed to leading order in \cite{DMS-2d,MS-2d} (thermodynamic limit) and in \cite{DS} (Gross-Pitaevskii regime).

\item It is interesting to compare our bound (\ref{eq:EN}) with the 
second order approximation of the energy per particle in the thermodynamic limit, given by
\be \label{eq:e0rho}
e (\r) = 4 \pi \r \, b \Big( 1 - b  | \log b | +  \Big( \frac 12 + 2\g + \log \pi \Big)  b + o(b)\Big)
\ee
with $b= |\log (\r \aa^2)|^{-1}$ and where $\g = 0.577..$ is Euler's constant.  
This expression, first predicted in  \cite{GroundState2d-1,GroundState2d-2,GroundState2d-3}, has been recently proved, for all positive potentials with finite scattering length, in \cite{FGJMO} (partial results have been previously obtained in \cite{FNRS-2d}, restricting the analysis to quasi-free states). In the Gross-Pitaevskii limit, where $\rho = N$ and $b = (2N - \log N - \log \frak{a}^2)^{-1}$, one can check that (\ref{eq:e0rho}) is consistent with (\ref{eq:EN}) (in the thermodynamic limit, the lattice spacing in $\Lambda^*$ tends to zero, and the sum over $p \in \Lambda^*$ is replaced by an integral, which is convergent because $\lim_{r \to 0} J_0(r)=1$).

\item \label{rmk4} It is interesting to observe that \eqref{eq:EN} and \eqref{eq:exc-main} only depend on the interaction potential through  the term $\pi^2 \frak{a}^2$ and the argument of the Bessel function $J_0$ in the expression for the ground state energy. Observing that the quantity 
\[ - 2\pi \log (\ell / \frak{a}) + \pi^2 \ell^2 - 4 \pi^2 \sum_{p \in \L^*_+} J_0(\ell |p|)/p^2 \]
is independent of the choice of $\ell > 0$ (see the end of the proof of Prop. \ref{prop:gsandspectrum}), we could also rewrite 
\[ E_N = 2\pi (N-1)  + 2\pi \log \frak{a}   + \pi^2 + \frac{1}{2} \sum_{p \in \Lambda^*_+} \Big[ \sqrt{|p|^4+8\pi p^2}-p^2-4\pi +\frac{(4\pi)^2}{2p^2} \big(1- J_0(|p|) \big)\Big]  \]
up to errors $\cO (N^{-1/10+ \delta})$, making the logarithmic form of the dependence of $E_N$ on $\frak{a}$ explicit. 
If we replace the potential $e^{2N}V(e^N .)$ in \eqref{eq:Ham0} with an interaction having scattering length $\aa^{(R)}_{N} = e^{-N/R} \aa$, for some $R > 0$, the estimates in Theorem \ref{thm:main} would depend on $R$; in particular \eqref{eq:EN} would read
\[ \begin{split} \label{eq:EN-R} 
E_N^{(R)} =\; & 2 \pi R (N-1) +  \pi R^2 \log R + \pi^2 \frak{a}^2 R  \\ 
& + \fra 12 \sum_{p\in 2\pi \ZZZ^2 \setminus\{0\}} \Big(\sqrt{p^4+8\pi R p^2}-p^2-4\pi R +\frac{(4\pi R)^2}{2p^2} \big(1- J_0(|p|\frak{a}/\sqrt{R}) \big)\Big)
\end{split}\]
and the dispersion of the low-energy excitations would be given by $\eps^{(R)}(p)= \sqrt{p^4 +8 \pi R \,p^2}$. Approximating the sum with an integral (in the limit of large $R$, after replacing the variable $p$ with $p/\sqrt{R}$), this leads to 
\[ E_N^{(R)} \simeq 2\pi R N + \pi R^2 \Big[ \frac{1}{2} + 2\gamma + \log \big(\pi R \frak{a}^2/2 \big) \Big] \]
which is perfectly consistent with the formula (\ref{eq:e0rho}) obtained in the thermodynamic limit.

\item The assumptions on $V$ are technical; the result is expected to hold true for any positive interaction with finite scattering length 
(in particular bounds compatible with \eqref{eq:EN} and upper bounds matching \eqref{eq:exc-main} for hard core interactions can be obtained  following \cite{FGJMO, BCOPS})
and also, more generally, for a certain class of (not necessarily non-negative) potentials having positive scattering length. The condition $V \in L^3 (\bR^2)$ is used to show some properties of the solution of the scattering equation, in Lemma \ref{lm:propomega}. The restriction to $V \geq 0$ is used to discard certain error terms, when proving lower bounds for the eigenvalues of (\ref{eq:Ham0}). 
\end{enumerate}


The proof of Theorem \ref{thm:main} is based on Fock space methods, recently developed in the three-dimensional setting, to study the dynamics of Bose-Einstein condensates \cite{BDS,BS}  and to investigate the equilibrium properties of dilute gases in the Gross-Pitaevskii regime. In particular, these techniques led to the verification of the predictions of Bogoliubov theory for the ground state energy and the excitation spectrum of three dimensional Bose gas in the Gross-Pitaevskii regime, confined on the unit torus \cite{BBCS4,HST} or by more general trapping potentials \cite{BSS1,NT}.

The starting observation is that, in order to investigate the low-energy properties of Bose gases, 
it is convenient to factor out the Bose-Einstein condensate and to focus on its orthogonal excitations. This suggests to introduce a unitary transformation $U_N$, mapping the $N$-particle 
Hilbert space $L^2_s (\Lambda^{N})$ into the truncated bosonic Fock space 
\be \label{eq;defcF} 
\cF^{\leq N}_+ = \bigoplus_{n\geq0}^N L^2_\perp (\Lambda^n)=  \bigoplus_{n\geq0}^N L^2_\perp (\Lambda)^{\otimes_sn} 
\ee
constructed over the orthogonal complement $L^2_\perp (\L)$ of the condensate wave function $\ph_0$ (defined by $\ph_0 (x) = 1$, for all $x \in \Lambda$). On the Hilbert space $\cF_+^{\leq N}$ we introduce the excitation Hamilton $\cL_N = U_N H_N U_N^*$, given by the sum of a constant and of terms that are quadratic, cubic and quartic in (appropriately defined) modified creation and annihilation operators (see \eqref{eq:calL} below). In the very spirit of the Bogoliubov approximation, we aim at reducing $\cL_N$ to a quadratic (and therefore diagonalizable) Hamiltonian, up to error terms vanishing in the limit of large $N$. To achieve this goal, we conjugate $\cL_N$ with suitable unitary operators, modelling the correlation structure created by the singular two-body interaction. 

The main input for our analysis are the recent results of \cite{CCS}, proving a bound of the form 
\begin{equation}\label{eq:rough} 2\pi N -C \le E_N \le 2\pi N + C \log N \end{equation} 
for the ground state energy and, most importantly, showing that the ground state and low-energy states of (\ref{eq:Ham0}) exhibit complete Bose-Einstein condensation, with at most order $\log N$ excitations. 
This estimate is used here to show that several error terms, emerging from the unitary conjugations can be neglected. 

While this strategy is similar to the one used in the three-dimensional setting (see, for example, \cite{BBCS2,BBCS4,HST,NT,BSS1,BCaS}), the choice of the appropriate unitary transformations and their action strongly depend on the specific problem under consideration. 

Compared with the three-dimensional setting, a first important difference we have to face to prove Theorem \ref{thm:main} is the fact that, in the two-dimensional Gross-Pitaevskii regime, correlations among particles are much stronger. This can already be seen by noticing that the expectation of (\ref{eq:Ham0}) on factorized states is of the order $N^2$, in the limit of large $N$. Hence, correlations among particles are responsible for reducing the ground state energy of (\ref{eq:Ham0}) to a quantity of order $N$. As a consequence, some additional care is required when studying the action of quadratic and cubic transformations that generate the correlation structure characterizing low-energy states. In particular, since cubic terms in the Hamilton operator carry large contributions to the energy (growing with $N$, as $N \to \infty$) we are not able to prove a-priori bounds on moments of the number of excitations (nor on products of the energy with moments of the number of excitations operator), which were important in the three dimensional setting \cite{BBCS4}. To overcome this problem, we are going to apply a localization on the number of particle argument (similarly to the one recently exploited in \cite{NT,HST}), combined with a-priori bounds on the energy of the excitations. A second important difference, compared with the three-dimensional setting, is that even after quadratic and cubic conjugations, the quartic part $\cV_N$ of the (renormalized) excitation Hamiltonian is not negligible on uncorrelated states. While this is not a problem for the derivation of lower bounds ($\cV_N$ is the restriction of the potential energy on the orthogonal complement of $\ph_0$; therefore, it is non-negative), it affects the proof of upper bounds for the eigenvalues of $H_N$. To circumvent this problem, we need to implement an additional unitary transformation, defined by the exponential of a quartic expression in creation and annihilation operators. Through this quartic conjugation, we eliminate the low-momentum part of $\cV_N$. This allows us to show upper bounds for the ground state energy and for low-energy excited eigenvalues of $H_N$ using uncorrelated states with low-momenta. This part is the main novelty of our work. We remark that unitary operators given by the exponential of quartic expressions in creation and annihilation operators have already been used in three dimensions in \cite{ABS}. The action of the quartic operators used here, however, is quite different. In particular, they renormalize the interaction up to contributions which are only negligible on suitable low-momentum states (we will use such low-momentum states as trial states, to prove upper bounds on the eigenvalues of (\ref{eq:Ham0})).

The plan of the paper is as follows. In the next section we introduce the formalism of second quantization and the map $U_N$, factoring out the condensate. Moreover, we define the quadratic transformation $e^B$ and the cubic transformation $e^A$ that allow us to approximate the renormalized excitation Hamiltonian $\cR_N = e^{-A} e^{-B} \cL_N e^B e^A$ by the sum of a quadratic Hamiltonian and of the quartic term $\cV_N$. The action of the unitary operators $e^B, e^A$, the properties of $\cR_N$ and their implications for Bose-Einstein condensation in low-energy states of (\ref{eq:Ham0}) are discussed in Section \ref{sec:renormalized}. Up to this point the analysis is similar to \cite{CCS} (some adaptation is still required, because we need here slightly stronger bounds, compared with those established in \cite{CCS}; for example we need an estimate for the energy of excitations, not only for their number). The real novelty of the present paper is in the Sects. \ref{sec:UB}-\ref{sec:proof}, where we show how to extract order one contributions to the ground state energy (to go from (\ref{eq:rough}) to the much more precise estimate (\ref{eq:EN})) and to compute low-energy excitations. In Sect.\,\ref{sec:UB} we introduce the quartic conjugation $e^D$ and we show how it can be used to get rid of the low-momentum part of $\cV_N$. In Sect.\,\ref{sec:diag}, we diagonalize quadratic Hamiltonians that have been derived in Sect.\,\ref{sec:renormalized} and in Sect.\,\ref{sec:UB} (we will work with two different quadratic Hamiltonians, one for the upper bounds, one for the lower bounds). The results from Sect.\,\ref{sec:renormalized}--\ref{sec:diag} are combined in Sect.\,\ref{sec:proof} to complete the proof of Theorem \ref{thm:main}; for the proof of the lower bounds, we apply here a localization argument.

\medskip

{\it Acknowledgements.} B.S. gratefully acknowledges partial support from the NCCR SwissMAP, from the Swiss National Science Foundation through the Grant ``Bogoliubov theory for bosonic systems'' and from the European Research Council through the ERC-AdG CLaQS. C.C. and S.C. warmly acknowledge  the GNFM Gruppo Nazionale per la Fisica Matematica - INDAM.

\section{The Renormalized Excitation Hamiltonian} \label{sec:renormalized}

We are going to describe excitations of the Bose-Einstein condensate on the truncated Fock space $\cF_+^{\leq N} = \bigoplus_{n=0}^N L^2_{\perp \ph_0} (\Lambda)^{\otimes_s n}$ constructed on the orthogonal complement of the zero-momentum orbital $\ph_0 (x) = 1$ for all $x \in \Lambda$. As first observed in \cite{LNSS}, we can define a unitary map $U_N : L^2_s (\Lambda^N) \to \cF_+^{\leq N}$ by requiring that $U_N \psi_N = \{ \alpha_0, \alpha_1, \dots , \alpha_N \}$, with $\alpha_j \in L^2_\perp (\Lambda)^{\otimes_s j}$ for all $j=0 ,1, \dots , N$, if 
\[ \psi_N= \alpha_0 \ph_0^{\otimes N} + \alpha_1 \otimes_s \ph_0^{\otimes (N-1)} + \dots + \alpha_N\, .\]
By definition, $U_N \psi_N \in \cF_+^{\leq N}$ describes the orthogonal excitations of the condensate, in the many-body state $\psi_N$. 

For any $p,q \in \Lambda^*_+ = 2\pi \bZ^2 \backslash \{ 0 \}$, we find (see \cite[Prop. 4.2]{LNSS}) 
\begin{equation}\label{eq:U-rules}
\begin{split} 
U_N \, a^*_0 a_0 \, U_N^* &= N - \cN_+   \\  
U_N \, a^*_p a_0 \, U_N^* &= a^*_p \sqrt{N-\cN_+ } =: \sqrt{N} b_p^* \\		
U_N \, a^*_0 a_p \, U_N^* &= \sqrt{N-\cN_+ } \, a_p =: \sqrt{N} b_p \\
U_N \, a^*_p a_q \, U_N^* &= a^*_p a_q    \end{split}
\end{equation} 
where $\cN_+ = \sum_{p\in \Lambda^*_+} a_p^* a_p$ is the number of particles operator on $\cF_+^{\leq N}$ and where we introduced modified creation and annihilation operators $b_p^*, b_p$ on $\cF_+^{\leq N}$, satisfying  
\begin{equation}\label{eq:bcomm} \begin{split} [ b_p, b_q^* ] &= \left( 1 - \frac{\cN_+}{N} \right) \delta_{p,q} - \frac{1}{N} a_q^* a_p 
		\\ [ b_p, b_q ] &= [b_p^* , b_q^*] = 0 \,\;\\
		[b_p, a^*_q a_r] &= \d_{p,q} b_r \,\qquad [b^*_p, a^*_q a_r] = - \d_{p,r} b^*_q
\end{split} \end{equation}
for all $p,q \in \Lambda^*$.

With $U_N$, we define the excitation Hamiltonian $\cL_N := U_N H_N U_N^*$, acting on a dense subspace of $\cF_+^{\leq N}$. Expressing \eqref{eq:Ham0} in second quantized form and using (\ref{eq:U-rules}), we find 
\begin{equation}\label{eq:calL}
	\cL_N 
=  \cL^{(0)}_{N} + \cL^{(2)}_{N} + \cL^{(3)}_{N} + \cL^{(4)}_{N}\end{equation}
where
\begin{equation}\label{eq:cLNj} \begin{split} 
		\cL_{N}^{(0)} =\;& \frac 12 \widehat{V} (0) (N-1)(N-\cN_+ ) + \frac 12 \widehat{V} (0) \cN_+  (N-\cN_+ ) \\
		\cL^{(2)}_{N} =\; &\cK + N\sum_{p \in \Lambda_+^*} \widehat{V} (p/e^N) \left[ b_p^* b_p - \frac{1}{N} a_p^* a_p \right] \\ &+ \frac N2 \sum_{p \in \Lambda^*_+} \widehat{V} (p/e^N) \left[ b_p^* b_{-p}^* + b_p b_{-p} \right] \\
		\cL^{(3)}_{N} =\; & \sqrt{N} \sum_{p,q \in \Lambda_+^* :\, p+q \not = 0} \widehat{V} (p/e^N) \left[ b^*_{p+q} a^*_{-p} a_q  + a_q^* a_{-p} b_{p+q} \right] \\
		\cL^{(4)}_{N} =\; & \cV_N  \, .
\end{split} \end{equation}
Here, we defined the Fourier transform of $V$ by  
\[ \widehat{V} (k) = \int_{\bR^2} V (x) e^{-i k \cdot x} dx \] 
for all $k \in \bR^2$, and we introduced the notation 
\be \label{eq:cKVN}
\cK = \sum_{p \in \Lambda^*_+} p^2 a_p^* a_p , \qquad \cV_N = \frac 12\sum_{\substack{p,q \in \Lambda_+^*, r \in \Lambda^*: \\ r \not = -p,-q}} \widehat{V} (r/e^N) a^*_{p+r} a^*_q a_p a_{q+r} \ee
for the kinetic and potential energy operators, restricted to the orthogonal complement of the condensate wave function. In the rest of the paper we are going to use the notation $ \cH_N = \cK + \cV_N$. 

The Hamilton operator $\cL_N$ is the starting point for our analysis. As discussed in the introduction, we are going to conjugate $\cL_N$ by suitable unitary operators to extract large contributions to the energy that are still hidden in $\cL_N^{(3)}, \cL_N^{(4)}$. To construct these unitary operators, we consider the ground state solution $f_\ell$ of the eigenvalue problem    
\be \label{fell-1}
\Big( -\D + \frac 12 V(x) \Big)  f_{\ell}(x) =  \l_{\ell}\, f_{\ell}(x)
\ee
on the ball $|x| \leq e^N \ell$, satisfying Neumann boundary conditions and normalized so that $f_{\ell}(x) = 1$ for $|x|= e^N\ell$ (for simplicity we omit here the $N$-dependence in the notation for $f_\ell$ and for $\lambda_\ell$). We will later choose $\ell = N^{-\alpha}$ with $\a >0$ so that $e^{-N} \ll \ell \ll 1$. The next Lemma (proven in Appendix \ref{App:omega})  collects properties of $f_\ell, \lambda_\ell$ that will be important for our analysis. 
\begin{lemma} \label{lm:propomega}
	Let $V\in L^3(\bR^2)$ be non-negative, compactly supported (with range $R_0$) and spherically symmetric, and denote its scattering length by $\aa$. For any $0<\ell<1/2$, $N$ sufficiently large, let $f_{\ell}$ denote the solution of  \eqref{fell-1}.  Then
	\begin{enumerate}[i)]
		\item 
		\be\label{eq:fell01} 
		0 \leq f_{\ell}(x) \leq 1   \qquad \forall \, |x| \leq e^N \ell\,.
		\ee 
		\item We have
		\be \label{eq:eigenvalue}
		\left |\l_{\ell} - \frac{2}{(e^N\ell)^2 \log(e^N\ell/\aa)}\left( 1 + \frac{3}{4}\fra{1}{\log(e^N\ell/\aa)} \right)  \right| \leq  \frac{C}{(e^N\ell)^2} \frac{1}{\log^3(e^N\ell/\aa)}\,. 
		\ee
		\item There exist a constant $C>0$ such that
		\be \label{eq:intpotf}  
		\left| 
		\int \di x\, V(x) f_{\ell}(x) - \frac{4\pi}{\log(e^N\ell/\aa)}\left( 1 + \fra{1}{2\log(e^N\ell/\aa)} \right)   \right| \leq   \frac{C}{\log^3 (e^N\ell/\aa)}\,. 
		\ee
		\item Let $w_\ell = 1- f_\ell$. Then there exists a constant $C>0$ such that 
		\[ \begin{split} \label{eq:boundsomega}
			|w_{\ell}(x)| &\leq  
                 \left\{ \begin{array}{ll} C \qquad &\text{if } |x| \leq R_0 \\ 
				C  \, \frac{\log(e^N\ell/|x|) }{\log (e^N\ell/\aa)} \hskip 1cm & \text{if } R_0 \leq |x|\leq e^N \ell    \end{array} \right. \\
			|\nabla w_{\ell}(x)| &\leq \frac{C}{\log (e^N\ell /\aa)} \frac 1 {|x| + 1}  \qquad \text{for all } \, |x| \leq e^N \ell\,.
		\end{split}\]
\end{enumerate}
\end{lemma}

We rescale the solution of (\ref{fell-1}), setting $f_{N,\ell} (x) := f_\ell (e^N x)$ for $|x| \leq \ell$, and $f_{N,\ell} (x) = 1$ for $x \in \Lambda$, with $|x| > \ell$. Then 
\begin{equation}\label{eq:scat2} \left( -\Delta + \frac{e^{2N}}{2} V (e^N x) \right) f_{N,\ell} (x) = e^{2N} \lambda_\ell f_{N,\ell} (x) \chi_\ell (x) \end{equation} 
with $\chi_\ell$ denoting the characteristic function of the ball $|x| \leq \ell$. Setting $w_{N,\ell} = 1 - f_{N,\ell}$, we find $w_{N,\ell} (x) = w_\ell (e^N x)$, if $|x| \leq \ell$, and $w_{N,\ell} (x) = 0$, if $x \in \Lambda$ and $|x| \geq \ell$ (recall, from Lemma \ref{lm:propomega}, that $w_\ell = 1-f_\ell$). We can then define $\check{\eta} : \L \to \bR$ as
\be \label{eq:defchecketa}
\check{\eta}(x)= - N w_{N,\ell}(x) = -N w_\ell(e^N x)  \,,
\ee
with Fourier coefficients \be \label{eq:defeta}
\eta_p = -N \widehat{w}_{N,\ell} (p) = - N e^{-2N} \widehat{w}_\ell(p/e^N)\,.
\ee
Notice that $\eta_p \in \bR$ (from the radial symmetry of $f_{\ell}$). To express the scattering equation (\ref{eq:scat2}) in terms of the coefficients $\eta_p$, it is useful to introduce the function $\omega_N \in L^\infty (\Lambda)$, defined through the Fourier coefficients 
\begin{equation}  \label{eq:defomegaN}
 \widehat \o_N(p)=2 N e^{2N} \l_\ell \widehat \chi_\ell(p) =  g_N \,  \widehat{\chi}(\ell p)\,, \qquad  g_N =2  N e^{2N}\ell^2 \l_\ell
\end{equation}
for all $p \in \L^*_+$ (here $\widehat \chi_\ell(p)$ and $\widehat \chi(p)$ denote the Fourier coefficients of the characteristic functions of the ball of radius $\ell$ and one respectively, and we used that $\widehat{\chi}_\ell(p) =\ell^2\widehat{\chi}(\ell p)$). Again, we find $\widehat{\o}_N (p) \in \bR$ (by radial symmetry of $\chi_\ell$). In the next lemma, we list some properties of $\check{\eta}$ 
and of $\o_N$. 
\begin{lemma} \label{lm:eta}
Let $V\in L^3(\bR^2)$ be non-negative, compactly supported  and spherically symmetric, and denote its scattering length by $\aa$. For any $0<\ell<1/2$, $N$ sufficiently large, let $\check{\eta}$ and $\o_N$ be defined as in  \eqref{eq:defchecketa} and \eqref{eq:defomegaN}, respectively. Then, we have $|\eta_0| \leq C \ell^2$ and 
$\widehat{\omega}_N (0) = \pi g_N$ with $|g_N| \leq C$, uniformly in $N$. More precisely, we find
\begin{equation}\label{eq:omegahat0} \big| \widehat{\omega}_N (0) - N \| V f_\ell \|_1 \big| \leq C N^{-1}\,. \end{equation} 
Moreover, we have $\widehat{\o}_N (p) \geq 0$ for all $p \in \L^*_+$ with $\ell |p| \leq 1$ and the pointwise bounds 
\[  |\eta_p| \leq \frac{C}{p^2}, \qquad |\widehat{\o}_N (p)| \leq C \min \Big\{ 1 , \frac{1}{(\ell |p| )^{3/2}} \Big\} \]
for all $p \in \L_+^*$. We also have the estimates   
\[
\| \eta \|_2^2 = \| \check{\eta} \|^2 \leq C \ell^2, \qquad  \| \check{\eta}\|^2_{H_1} \leq C N\,.
\]
Finally, for every  $p \in \L^*_+$, we can write (\ref{eq:scat2}) as  
\[ p^2 \eta_p + \frac{N}{2} \big( \widehat{V} (. /e^N) * \widehat{f}_{N,\ell} \big) = \frac{1}{2} \big( \widehat{\o}_N * \widehat{f}_{N,\ell} \big)  \]
or, equivalently, as 
\begin{equation}\label{eq:eta-scat}
	\begin{split} 
		p^2 \eta_p + \frac{N}{2} \widehat{V} (p/e^N) & + \frac{1}{2} \sum_{q \in \Lambda^*} \widehat{V} ((p-q)/e^N) \eta_q \\ &\hspace{2cm} = \frac{1}{2} \widehat{\o}_N (p) + \frac{1}{2N} \sum_{q \in \L^*} \widehat{\o}_N (p-q) \eta_q \,.
\end{split} \end{equation}
\end{lemma}

\begin{proof}  The bounds for $|\eta_0|$, $|\eta_p|$, $\| \eta \|_2$, $\| \check{\eta} \|_{H^1}$ have been established in \cite[Sect.\,3]{CCS}. The bounds for $\widehat{\o}_N (0)$ are a direct consequence of  Lemma \ref{lm:propomega} (in particular, of parts (ii) and (iii)). To prove that $\widehat{\o}_N (p) \geq 0$ for $p \in \L^*_+$ with $\ell |p| \leq 1$ and to show the estimate for $|\widehat{\o}_N (p)|$, we observe that, denoting by $J_1$ the Bessel function of the first kind of order $1$,  
\be \label{eq:chiellp}
 \widehat{\chi}_\ell (p) = \ell^2 \widehat{\chi} (\ell p) = 2\pi \ell \frac{J_1 (\ell |p|)}{|p|} \,,
\ee
From $0 \leq J_1 (r) \leq C r$ for all $0 \leq r \leq 2$, $|J_1 (r)| \leq C r^{-1/2}$ for all $r \geq 1$, we obtain the claim. 
\end{proof} 

As mentioned above, we choose $\ell = N^{-\alpha}$ so that $\| \eta \|^2, |\eta_0| \leq C N^{-2\a}$ will be small factors. With the coefficients $\eta_p$, introduced in (\ref{eq:defeta})  we define, following \cite{CCS}, the antisymmetric operators 
\be\label{eq:defB} 
	B = \frac{1}{2} \sum_{p\in \L^*_+}  \eta_p \left( b_p^* b_{-p}^* - \hc \right) \, 
\ee
and
\begin{equation}\label{eq:defA} A = \frac1{\sqrt N} \sum_{p, v \in \L^*_+} 
	 \eta_p \big( b^*_{p+v}a^*_{-p}a_v - \text{h.c.}\big)\, . 
\end{equation}
We will consider the unitary operators $e^B$ and $e^A$. For our analysis, it will be important to control the growth of number of particles and energy with respect to the action of $e^B, e^A$; the following lemma is proven in \cite[Sect.\,3-4]{CCS}. 
\begin{lemma} \label{lm:ANgrow}
	Suppose that $B, A$ are defined as in (\ref{eq:defB}) and (\ref{eq:defA}). Then, for any $k\in \bN$ there exists a constant $C >0$ (depending on $k$) such that   
\[ e^{-B} (\cN_+ + 1)^k e^B , \; e^{-A} (\cN_++1)^k e^{A} \leq C (\cN_+ +1)^k\,.   \] 
Moreover, we also have the following bound for the growth of the energy w.r.t. $e^A$ (a similar estimate also holds for the action of $e^B$, but we will not need it in the sequel): 
\[ e^{-sA} \cH_N e^{sA} \leq C \cH_N + C N (\cN_+ +1) \]
holds true on $\cF_+^{\leq N}$, for any $\alpha > 0$ (recall the choice $\ell = N^{-\alpha}$ in the definition (\ref{eq:defeta}) of the coefficients $\eta_p$), for all $\a \geq 1$,  $s \in [0;1]$ and $N \in \bN$ large enough. 
\end{lemma}

With $A,B$, we define the renormalized excitation Hamiltonian 
\be \label{eq:calR}
\cR_{N}= e^{-A} e^{-B} U_N H_N U^*_N e^{B} e^A \,.
\ee
In the next proposition, we collect important properties of $\cR_N$. Part a) isolates the important contributions to $\cR_N$; its proof follows closely the proof of Prop. 4 in \cite{CCS} and is deferred to Appendix \ref{App:propRN}. Part b) and c), on the other hand, are consequences of part a) and will be used to show upper and, respectively, lower bounds on the eigenvalues. 
\begin{prop} \label{prop:RN} 
Let $V\in L^3(\bR^2)$ be compactly supported, pointwise non-negative and spherically symmetric. Let $\cR_{N}$ and $\widehat{\o}_N$ be defined in \eqref{eq:calR} and \eqref{eq:defomegaN}, respectively. Let $\ell =N^{-\a}$ and $\a \geq 5/2$. 
\begin{itemize}
\item[a)] There exists a constant $C>0$ such that 
	\begin{equation}\label{eq:cReff}
		\begin{split} 
			\cR_{N}= &\;  \frac N2 \big(\widehat{V}(\cdot/e^N)*\widehat{f}_{N,\ell}\big)(0) (N-1)\left(1-\frac{\cN_+}{N}\right)  + \frac 12\sum_{p\in \L^*_+} \widehat{\o}_N(p)\eta_p\\
		& + \frac N 2  \big(\widehat{V}(\cdot/e^N)*\widehat{f}_{N,\ell}\big)(0) \,\cN_+ \left( 1 - \frac{\cN_+}N \right) \\
			& +  \widehat \o_N(0) \sum_{p\in \L^*_+}a^*_pa_p \Big(1-\frac{\cN_+}{N} \Big)+  \frac 12 \sum_{p\in \L^*_+}  \widehat{\o}_N(p)\big[ b^*_p b^*_{-p} + b_p b_{-p} \big]  \\
			& +\frac 1 {\sqrt N} \sum_{\substack{r,v\in \L^*_+:\\ r\neq-v} } \widehat{\o}_N(r)\big[ b^*_{r+v}a^*_{-r} a_v + \text{h.c.}\big]  + \cH_N  + \cE_\cR
		\end{split}
	\end{equation} 
where 
	\be \label{eq:ReffE}
	\pm  \cE_\cR \leq C  N^{-1/2} (\log N)^{1/2} (\cH_N +1)  \, , 
	\ee
	for $N \in \bN$ sufficiently large.
\item[b)] Let
\be \label{eq:defC-cR}
C_{\cR}  =\frac N2 \big(\widehat{V}(\cdot/e^N)*\widehat{f}_{N,\ell}\big)(0) (N-1) + \frac 12\sum_{p\in \L^*_+} \widehat{\o}_N(p)\eta_p\,.
\ee
Let $P_L$ be the low-momenta set
\be \label{eq:defPL}
P_L=\{p\in \L^*_+\; : |p| \leq N^{\a+\nu}\}\,,
\ee 
with $\nu \in (0; 1/2)$.  Let 	
\begin{equation}
		\label{eq:defQNL}
		\begin{split}
			Q_{\cR}^{(L)} & := \sum_{p\in P_L} (p^2 + \widehat{\omega}_N (p)) b^*_p b_p   +  \frac 12 \sum_{p\in P_L} \widehat{\omega}_N (p) \big[ b^*_p b^*_{-p} + b_p b_{-p} \big]\,.  
		\end{split}
	\end{equation}
	Then
\begin{equation}\label{eq:prRN-2}
\cR_{N}= C_{\cR} + Q^{(L)}_{\cR} + \sum_{p \in P_L^c} p^2 a_p^* a_p + \cV_N + \cE'_{\cR},\\
\end{equation} 
for an error term $\cE'_{\cR}$ satisfying 
\be \label{eq:cE-cRprimo}
\pm \cE_{\cR}' \leq C \left[ N^{-3\nu/2} + N^{- 1/2} (\log N)^{1/2} \right] (\cN_++1)(\cH_N+1) 
\ee
on $\cF_+^{\leq N}$.
\item[c)] Finally, let $\nu \in (1/6 ; 1/2)$ and $P_L$ as above; then there exists a constant $C$ such that for any $\g\in (0 ; 1/4)$ we have
\be \begin{split} \label{eq:prRN-3} 
\cR_N \geq & \;  C_{\cR} +  \sum_{p \in P_L} \Big((1 - C N^{-\g}) p^2 + \widehat \o_N(p)\Big) b^*_p b_p + \tfrac 1 2 N^{\g}\sum_{p \in \L^*_+\setminus P_L}a^*_pa_p \\
&+ \frac 12 \sum_{p\in P_L} \widehat{\omega}_N (p) \big[ b^*_p b^*_{-p} + b_p b_{-p} \big]  
- C (\log N)N^{\g-1}  (\cN_++1)^2  -  C  N^{-\g} \,.
\end{split}
\ee
\end{itemize} 
\end{prop}

{\it Remark.} Conjugating with the unitary operators $e^B$ and $e^A$, we effectively replace the interaction $\widehat{V} (p/e^N)$ appearing in (\ref{eq:cLNj}) with the renormalized potentials $(\widehat{V} (./e^N) * \widehat{f}_{N,\ell})$ and $N^{-1} \widehat{\omega}_N$. More precisely, conjugation with $e^B$ renormalizes the off-diagonal quadratic term (second term on the third line of (\ref{eq:cReff})), while the cubic conjugation renormalizes the diagonal quadratic and the cubic terms. Renormalization arises when combining terms in $\cL_N$ with contributions from the commutators $[B,\cL_N]$ and $[A,\cL_N]$. At the same time, this procedure produces new constant terms, reducing $\cL_N^{(0)}$ in (\ref{eq:cLNj}) (a term of order $N^2$) to the first line in (\ref{eq:calR}) (order $N$). After renormalization with $e^B$ and $e^A$, the only term in (\ref{eq:cReff}) still depending on the original potential $\widehat{V} (p/e^N)$ is the quartic term $\cV_N$ (contained in $\cH_N = \cK + \cV_N$). In contrast with the three-dimensional setting, $\cV_N$ is here of order one (on uncorrelated trial states); this is the reason why, to show upper bounds on the eigenvalues of $\cR_N$, we will need an additional conjugation, with a quartic phase.

\begin{proof}[Proof of Prop. \ref{prop:RN}]
As explained above, the proof of part a) is sketched in App. \ref{App:propRN}. 

Part b) follows from part a). In fact, the cubic term appearing on the r.h.s. of \eqref{eq:cReff} can be estimated by 
\be \begin{split} \label{eq:RN-LB3}
		\Big|  \frac 1 {\sqrt N} & \sum_{\substack{r,v\in \L^*_+\\ r \neq -v}}  \widehat{\o}_N(r) \langle \xi, b^*_{r+v} a^*_{-r} a_v\xi \rangle \Big|  \\ & \leq  \frac 1 {\sqrt N} \sum_{\substack{r,v\in \L^*_+\\ r \neq -v}}|\widehat{\o}_N(r) | \| (\cN_++1)^{-1/2}b_{r+v} a_{-r} \xi \|  \|(\cN_++1)^{1/2} a_v \xi \| \\
		&   \leq  \frac 1 {\sqrt N} \, \bigg[ \sum_{\substack{r,v\in \L^*_+\\ r \neq -v}} |r|^2 \| (\cN_++1)^{-1/2}b_{r+v} a_{-r} \xi \|^2 \bigg]^{1/2} 
		\\& \hskip 2cm \times 
		\bigg[ \sum_{\substack{r,v\in \L^*_+\\ r \neq -v}} \frac{|\widehat{\o}_N(r) |^2}{|r|^2}   \| (\cN_++1)^{1/2}a_v \xi \|^2 \bigg]^{1/2}\\
		& \leq  \frac{C (\log N)^{1/2}}{\sqrt N}  \,\| \cK^{1/2}\xi \| \| (\cN_+ +1)\xi\|\,,
	\end{split}\ee
where we used that 
\be \label{eq:omegasup2}
\sum_{p \in \L^*_+} \frac{|\widehat{\o}_N(p)|}{|p|^2} \leq C \log N \,.
\ee
Moreover, we can write 
\be \label{eq:calE-1}
 \sum_{p \in P_L} p^2 a_p^* a_p = \sum_{p \in P_L} p^2 b_p^* b_p + \cE_1 
 \ee
and 
\be  \label{eq:calE-2}
\widehat{\omega}_N (0) \sum_{p \in \L^*_+} a_p^* a_p \left( 1- \frac{\cN_+}{N} \right) = \sum_{p \in P_L} \widehat{\omega}_N (p) b_p^* b_p + \cE_2 \ee
for error terms $\cE_1, \cE_2$ satisfying 
\be \label{eq-calE-12}
 \pm \cE_1, \cE_2 \leq C N^{-1} (\cK + 1) (\cN_+ + 1) 
 \ee
for all $\alpha \geq 1$ (here we used $|\widehat{\omega}_N (p) - \widehat{\omega}_N (0)| \leq C |p|/N^\alpha$ and also the bound $\widehat{\omega}_N (p) \leq C$, to control the contribution from $|p| > N^{\alpha + \nu}$). 

As for the off-diagonal quadratic contribution associated with momenta $p \in P_L^c$, we find, with Lemma \ref{lm:eta}, 
\[ \label{eq:sumomegap2high}
\sum_{p\in P_L^c}\fra{|\widehat\o_N(p)|^2}{p^2}  \leq C\, \sum_{p\in P_L^c}\fra{N^{3\a}}{|p|^5}\leq C N^{-3 \nu}\,.
\]
Hence  
\be \begin{split} \label{eq:off-diag-cR}
|\langle \xi, \sum_{p \in P_L^c} \widehat \omega_N(p) b_p b_{-p}\xi\rangle| &  \leq C \, \Big[ \sum_{p \in P_L^c} \frac{|\widehat \omega_N(p)|^2}{|p|^2} \Big]^{1/2} \Big[\sum_{p \in P_L^c} p^2 \| b_p\xi\|^2\Big]^{1/2} \|(\cN_++1)^{1/2}\xi\| \\ &\leq C N^{-3\nu/2} \| \cK^{1/2} \xi \| \| (\cN_+ + 1)^{1/2} \xi \|\,. 
\end{split}\ee
From Eq. \eqref{eq:RN-LB3} and Eqs.\eqref{eq:calE-1}--\eqref{eq:off-diag-cR}, together with the simple bound 
\[
\Big| \frac N2 \big(\widehat{V}(\cdot/e^N)*\widehat{f}_{N,\ell}\big)(0) \Big[ (N-1) \cN_+ /N- \cN_+ \Big(1 - \cN_+/N\Big) \Big]\Big| \leq \frac{C}{N} \cN^2_+  \, ,  
\]
 we obtain \eqref{eq:prRN-2}. 
 
Finally, we show part c). Again, we start from (\ref{eq:cReff}) and we use (\ref{eq:RN-LB3}) to bound the cubic term and (\ref{eq:off-diag-cR}) to control the off-diagonal  quadratic contribution associated with $p \in P_L^c$. Instead of (\ref{eq:calE-2}), we notice that, since $\widehat \o_N (0) > 0$,
 \[
 \sum_{p \in \L^*_+} \widehat \o_N(0) a^*_p a_p \Big( 1 - \frac{\cN_+}N \Big) \geq   \sum_{p \in P_L} \widehat \o_N(0) b^*_p b_p - \frac C  N \,\cN_+\,.
 \]
This bound, combined with the observation that, by \eqref{eq:defomegaN}, 
 \[
\Big| \sum_{p \in P_L} \big(\widehat \o_N(p)- \widehat \o_N(0) \big) \langle \xi, b^*_p b_p \xi \rangle \Big| \leq C N^{-\a} \sum_{p \in P_L} |p| \| a_p \xi\|^2 \leq C N^{-\a} \| \cK^{1/2}\xi\|^2 
\]
and with $\cV_N \geq 0$ implies that, for any $\g >0$, 
\[ \begin{split}
\cR_N \geq &\; C_\cR +   \sum_{p \in P_L} \widehat \o_N(p) b^*_p b_p  + \frac 12 \sum_{p\in P_L} \widehat{\omega}_N (p) \big[ b^*_p b^*_{-p} + b_p b_{-p} \big] \\
&  + \cK   -  C \Big[ N^{-\g} + N^{-3\nu/2} + N^{-1/2}(\log N)^{1/2}  \Big] (\cK+1)  - \frac{\log N}{N^{1-\g}} \,(\cN_++1)^2\,,
\end{split}\] 
Later on we will need to fix $\g <1/4$ to control the error proportional to $\cN_+^2$. With this restriction and for $\nu \in (1/6 ; 1/2)$ there exists $C$ such that
\[ \begin{split}
\cR_N \geq &\; C_\cR +  \sum_{p \in P_L} \widehat \o_N(p) b^*_p b_p  + \frac 12 \sum_{p\in P_L} \widehat{\omega}_N (p) \big[ b^*_p b^*_{-p} + b_p b_{-p} \big]   + \tfrac 1 2 N^{\g} \hskip -0.2cm\sum_{p \in \L^*_+ \setminus P_L}a^*_pa_p\\
&  +(1- C N^{-\g}) \sum_{p \in P_L} p^2 a^*_p a_p - \frac{\log N}{N^{1-\g}} \,(\cN_++1)^2 - C N^{-\g}\,.
\end{split}\]
Here, we divided the kinetic energy into the sum of two operators; in the one associated with $p \in \L^*_+ \backslash P_L$ we estimated $p^2 \geq N^\gamma$. With $a_p^* a_p \geq b_p^* b_p$, we obtain (\ref{eq:prRN-3}). 
\end{proof}

As shown in \cite[Theorem 1.1]{CCS}, an important consequence of part a) of Prop. \ref{prop:RN} is the emergence of Bose-Einstein condensation for low-energy states, with an optimal control on the number of orthogonal excitations. This also implies an upper bound for the expectation of the operator $\cH_N$, on the excitation vectors associated with low-energy states; this is the content of the next proposition. 
\begin{prop}\label{lm:hpN}
	Let $V \in L^3(\bR^2)$ be non-negative, compactly supported and spherically symmetric. 	Let $\psi_N \in L^2_s (\Lambda^N)$ with $\| \psi_N \| = 1$ belong to the spectral subspace of $H_N$ with energies below $2\pi N + \zeta$, i.e. 
	\begin{equation} \label{eq:psi-space} \psi_N = {\bf 1}_{(-\infty ; 2\pi N + \zeta]} (H_N) \psi_N \, . 
	\end{equation} 
	Let $\xi_N = e^{-A} e^{-B} U_N \psi_N$ be the renormalized excitation vector associated with $\psi_N$. Then for any $\a \geq 5/2$, there exists a constant $C > 0$ such that 
	\begin{equation}\label{eq:hpN} \langle \xi_N, (\cH_N+1) \,\xi_N \rangle \leq C (1+\z)(\log N) \, . \end{equation}
\end{prop}
\begin{proof}
Combining the bounds \cite[Eqs.\;(58)-(59)]{CCS} for the excitation Hamiltonian $\cR^{\text{eff}}_{N,\a}$ defined in \cite[Eq.\;(47)]{CCS} with the estimate (\ref{eq:ReffE}) (and with the observation that, by \eqref{eq:omegahat0}, $|\cR^{\text{eff}}_{N,\a} - (\cR_N - \cE_\cR)| \leq C$), we conclude that 
\[ \cR_N \geq 2\pi N + \frac{1}{2} \cH_N - C (\log N) (\cN_+ + 1) \]
for any $\a \geq 5/2$ and $N$ large enough. The assumption (\ref{eq:psi-space}), and the definition of $\xi_N$ imply therefore that 
\[ \langle \xi_N, \cH_N \xi_N \rangle \leq 2 \zeta + C (\log N) \langle \xi_N, (\cN_+ + 1) \xi_N \rangle\,. \]
From the condensation estimate \cite[Eq. (61)]{CCS}  and from Lemma \ref{lm:ANgrow}, we conclude that 
\[  \langle \xi_N, \cH_N \xi_N \rangle  \leq  C (\zeta + 1)  (\log N)\,. \]
\end{proof}


\section{Quartic conjugation}
\label{sec:UB}

From (\ref{eq:prRN-2}), it is clear that to prove upper bounds on the eigenvalues of $\cR_N$, we cannot ignore the contributions of $\cV_N$ on the r.h.s. of (\ref{eq:prRN-2}). Instead, we conjugate $\cR_N$ with a quartic phase, which (up to errors that can be neglected) removes the low-momentum part of $\cV_N$, leaving us with an operator whose expectation vanishes on states generated by the action of creation operators $a_p^*$, with $p \in P_L$, the low-momentum set defined in (\ref{eq:defPL}). At the end, this will allow us to show upper bounds for the eigenvalues of $\cR_N$, making use of trial states involving only particles with low momentum. 

We consider the quartic operator 
\begin{equation}\label{eq:defD}
	D := D_1-D_1^*=\frac1{4N} \sum_{\substack{r \in \L^*_+, v,w \in P_L\\ v\neq -r, w\neq r} } 
	\eta_r \big[a^*_{v+r}a^*_{w-r}a_va_w - a^*_{w}a^*_{v}a_{w-r}a_{v+r}\big]\,,
\end{equation}
acting on $\cF^{\leq N}_+$.
Here, $\eta_p$ is defined as in \eqref{eq:defeta} and $P_L = \{ p \in \Lambda^*_+ :\; |p| \leq N^{\alpha + \nu} \}$. 

Since $D$ commutes with the number of particles operator $\cN_+$, we trivially obtain that 
\begin{equation*}\label{eq:DNgrow} e^{-D} (\cN_++1)^k e^{D}  = (\cN_+ +1)^k  \end{equation*}
for all $k \in \bN$. 

We state now two lemmas that will be shown in the next subsections. In the first lemma, we control the action of the quartic transformation on the kinetic energy operator. 
\begin{lemma}\label{lm:aprioriestimateKD} 
Let $\cK$  and $D$ be defined in \eqref{eq:cKVN} and in \eqref{eq:defD}, respectively, with $\a \geq 5/2$ and $\nu \in (0,1/2)$. Let $\k \in \bN$ the smallest integer s.t.  $\k > 4 (\a + \n -1/2 ) $. Then there exists $C > 0$ such that  
\begin{equation}\label{eq:DKD}
e^{-D} \cK e^{D} \leq C  \cK  \big(\cN_++1 \big)^{\k+2}  
	\end{equation}
for $N$ large enough. Moreover, we find 
\begin{equation}\label{eq:DKD-K} \pm \Big[ e^{-D} \cK e^D - \cK \Big] \leq C \frac{(\log N)^{1/2}}{N^{1/2}} \cK (\cN_+ + 1)^{\kappa +3} \, . \end{equation} 
 \end{lemma}
 {\it Remark.} Since $\cN_+$ commutes with $D$, (\ref{eq:DKD}) also implies that 
 \[ e^{-D} \cK (\cN_+ + 1)^j e^D \leq C \cK (\cN_+ + 1)^{\k + j +2} \]
 for all $j \in \bN$. 
 
In the second lemma, we bound the growth of the potential energy operator. 
\begin{lemma}
\label{lm:growthVN}
Let $\cV_N$ and $D$ be defined in \eqref{eq:cKVN} and \eqref{eq:defD}, respectively. Fix $\a \geq 5/2$,  $\nu \in (0,1/2)$,  and let $\k \in \bN$ be the smallest integer such that $\k > 4(\a+\nu-1/2)$. 
Recalling the definition (\ref{eq:defPL}) of the set $P_L$, let 
\begin{equation}\label{eq:defVNL}
\cV_N^{(L)}  =\; \fra{1}{4N}\sum_{\substack{u \in \L^*, v,w \in P_L\\ u\neq -v, w}} \widehat{V}(u/e^N)\big[a^*_{v+u}a^*_{w-u}a_va_w + \hc\big]
\end{equation}
and 
\be \label{eq:defcV-H}
\begin{split} 		\cV_N^{(H)} &= \cV_N - \cV_N^{(L)} \\ &= \fra 12 \sum_{\substack{r \in \L^*, v\in P_L^c\\ w\in P_L\\r\neq -v, w}} \hskip -0.2cm \widehat{V}(r/e^N)\; a^*_{v+r}a^*_{w-r}a_va_w  +\fra 14\sum_{\substack{r \in \L^*, v,w\in P_L^c\\r\neq -v, w}} \hskip -0.2cm \widehat{V}(r/e^N)\; a^*_{v+r}a^*_{w-r}a_va_w+\hc .
\end{split} 	\ee
Then, we have 
\begin{equation}\label{eq:DVD-VH} e^{-D} \cV_N e^D = \cV_N^{(H)} + \cE_{\cV_N} \end{equation} 
where
\[ \pm \cE_{\cV_N} \leq C N^{\nu - 1/2} (\log N)^{1/2} \cK (\cN_+ + 1)^{\k +4} \]
for all $N \in \bN$ large enough. 
\end{lemma} 
	
Using the last two lemmas, we can describe the action of the quartic transformation on the renormalized excitation Hamiltonian $\cR_N$. Our goal consists in proving that, on low-momentum states, the operator $e^{-D} \cR_N e^D$  is given, up to negligible errors, by a quadratic Hamiltonian which will be later diagonalized in Prop. \ref{prop:gsandspectrum}. 
\begin{prop}
 \label{prop:expD-cRN}
	Let $ \cR_{N}$  be defined as in \eqref{eq:prRN-2} and $D$ defined as in \eqref{eq:defD} with $\a \geq 5/2$ and  $\nu \in (0,1/2)$. Let  $C_\cR$ and $\cQ_\cR^{(L)}$ be defined in \eqref{eq:defC-cR} and \eqref{eq:defQNL}, respectively. Suppose that $\xi_L \in \cF_+^{\leq N}$ is such that $a_p \xi_L = 0$, for all $p \in P_L^c$, with the low-momentum set $P_L$ defined as in (\ref{eq:defPL}). Then, we have  
\begin{equation}\label{eq:expD}
\begin{split}
	\Big| \bmedia{\xi_L, e^{-D}\,\cR_{N} e^D \xi_L } 
		&-  \bmedia{ \xi_L, \big(C_{\cR} + \cQ_\cR^{(L)}\big) \xi_L} \Big| \\ &\leq C \Big[ N^{-3\nu/2} + N^{\nu-1/2} (\log N)^{1/2}\Big] \, \langle \xi_L, \cK \, (\cN+1)^{\k+5} \xi_L \rangle  \,,
\end{split}
\end{equation}
where $\k \in \bN$ is the smallest integer s.t. $\k > 4 (\a+\nu -1/2)$ and $N \in \bN$ is large enough.
\end{prop} 

\begin{proof}
From (\ref{eq:prRN-2}), we can write 
\[ \cR_N = C_\cR + \wt{\cQ}_{\cR}^{(L)} + \cK + \cV_N + \wt{\cE}_{\cR} \]
where \[ \wt{Q}_\cR^{(L)} := \cQ_\cR^{(L)}  - \sum_{p \in P_L} p^2b^*_pb_p= \sum_{p \in P_L} \widehat{\omega}_N (p) \Big[ b_p^* b_p + \frac{1}{2} \big( b_p^* b_{-p}^* + b_p b_{-p} \big) \Big] \]
and 
\[ \pm \wt{\cE}_\cR \leq C \left[ N^{-3\nu/2} + N^{- 1/2} (\log N)^{1/2} \right] (\cH_N+1)(\cN_++1) \, .\] 
With Lemma \ref{lm:aprioriestimateKD} and Lemma \ref{lm:growthVN}, we find immediately that 
\begin{equation}\label{eq:DED} | \bmedia{\xi_L, e^{-D} \wt{\cE}_{\cR} e^D \xi_L }| \leq C \big[ N^{-3\nu/2} + N^{-1/2} (\log N)^{1/2}\big] \langle \xi_L , \cK (\cN_++1)^{\k+5}\xi_L\rangle  \,.
\end{equation} 
Moreover, with (\ref{eq:DKD-K}) we obtain (using that $a_p \xi_L = 0$ for all $p \in P_L^c$) 
\be\label{eq:LDKDL} \Big| \langle \xi_L, e^{-D} \cK e^D \xi_L \rangle - \sum_{p \in P_L} p^2  \, \langle \xi_L,  a_p^* a_p  \xi_L \rangle \Big| \leq  C N^{-1/2} (\log N)^{1/2} \langle \xi_L, \cK (\cN_++1)^{\k+3}\xi_L \rangle \ee
and, with (\ref{eq:DVD-VH}), we find 
\be\label{eq:LDVDL}  \langle \xi_L, e^{-D} \cV_N e^{D} \xi_L \rangle \leq C N^{\nu-1/2} (\log N)^{1/2} 
\langle \xi_L, \cK (\cN_+ + 1)^{\k + 4} \xi_L \rangle \, . \ee 
It remains to study $e^{-D} \wt{\cQ}_{\cR}^{(L)} e^D$. Writing 
 $a^*_{v+r}a^*_{w-r}a_v a_w = a^*_{v+r}a_v a^*_{w-r}a_{w} - a^*_{v+r}a_w \d_{w, r+v}$ and using \eqref{eq:bcomm},  we find   
\[ [\wt{\cQ}_\cR, D]=  \sum_{i=1}^6 Z_i \]
with 
\[ \begin{split} 
Z_1 = &\; \fra{1}{2N}\; \sum_{\substack{r\in \L^*_+,\,v,w \in P_L\\r+v \in P_L,\,r\neq w}}  \widehat{\o}_N(v+r)\eta_r\,\big(b_{v+r}^*a^*_{w-r} a_w b_v + \hc \big)\\
Z_2 = &\; - \fra{1}{2N}\; \sum_{\substack{r\in \L^*_+\\v,w \in P_L\\r\neq -v,w }}  \widehat{\o}_N(v) \eta_r\,\big(b_{v+r}^*a^*_{w-r} a_w b_v + \hc\big) \\		
Z_3 = &\; \fra 1{2N}\sum_{\substack{r \in \L^*_+, v,w \in P_L\\ v+ r \in P_L,\, w\neq r}} \widehat{\o}_N(v+r) \eta_r\, ( b^*_{-r-v}b^*_v a^*_{w}a_{w-r} +\hc) \\
Z_4 = &\; \fra 1{ 4N}\sum_{\substack{r \in \L^*_+, v \in P_L\\ v+r \in P_L}} \widehat{\o}_N(v+r) \eta_r\, ( b_{v} b_{-v} +\hc) \\
Z_5 = &\; - \fra 1{2 N}\sum_{\substack{r \in \L^*_+, v,w \in P_L\\ r \neq -v, w}} \widehat{\o}_N(v) \eta_r\, ( b^*_{v+r}b^*_{-v} a^*_{w-r}a_{w} +\hc) \\
Z_6 = &\; - \fra 1{4 N}\sum_{\substack{r \in \L^*_+, v \in P_L\\ v \neq r}} \widehat{\o}_N(v) \eta_r\, ( b_{v-r} b_{-v+r} +\hc)\,.	
\end{split}\]

We can estimate 
\be \label{eq:Ki}
\pm Z_i \leq   C \frac{(\log N)^{1/2}}{N^{1/2}}  \cK (\cN_++1)\,.
\ee
for all $i=1, \dots , 6$. Indeed
	\[
	\begin{split}
 |\langle \xi, Z_1\;\xi\rangle|,  |\langle \xi, Z_2 \;\xi\rangle| &\leq \fra CN\,   \| \widehat{\o}_N \|_\infty \,\sum_{\substack{r\in \L^*_+\\v,w \in P_L\\r\neq -v,w }} |\eta_r|\|a_{v+r}a_{w-r}\xi\|\|a_va_w \xi\|\\
		& \leq \fra CN \; \Big[\sum_{\substack{r\in \L^*_+\\v,w \in P_L\\r\neq -v,w }}  \fra 1{|v|^2}\|a_{v+r}a_{w-r} \xi \|^2\Big]^{1/2}\Big[\sum_{\substack{r\in \L^*_+\\v,w \in P_L\\r\neq -v,w }} |\eta_r|^2 |v|^2\|a_{v}a_{w}\xi \|^2\Big]^{1/2}\\
		& \leq C N^{-1-\a}(\log N)^{1/2}\|\cN_+ \xi\|\|\cK^{1/2}(\cN_++1)^{1/2}\xi\|\,.
	\end{split}
	\]
	As for $Z_3$ we get
	\[
	\begin{split}
		|\langle \xi, Z_3\;\xi\rangle| 		&\leq \fra C{N} \,  \| \widehat{\o}_N \|_\infty \sum_{\substack{r \in \L^*_+, v,w \in P_L\\ v\neq -r, w\neq r}}|\eta_r|\|(\cN_++1)^{-1/2}b_{-v-r}b_va_{w}\xi\|\|(\cN_++1)^{1/2}a_{w-r}\xi\|\\
		&\leq \fra C{N}\Big[\sum_{\substack{r \in \L^*_+, v,w \in P_L\\ v\neq -r, w\neq r}}\fra 1{|v|^2}|\eta_r|^2\|a_{w-r}\cN_+^{1/2}\xi\|^2\Big]^{1/2}\\
		&\hspace{3cm}\times\Big[\sum_{\substack{r \in \L^*_+, v,w \in P_L\\ v\neq -r, w\neq r}} |v|^2\|(\cN_++1)^{-1/2}  b_{-v-r} b_v a_{w}\xi\|^2\Big]^{1/2}\\
		&\leq CN^{-\a-1}(\log N)^{1/2} \|\cN_+ \xi\|\|\cK^{1/2}(\cN_++1)^{1/2}\xi\|\,. 
	\end{split}
	\]
	Finally, using that $|\eta_r|\leq |r|^{-2}$, together with \eqref{eq:omegasup2}, we end up with
	\[
	\begin{split}
		|\langle \xi, Z_4\;\xi\rangle| &\leq \fra C{N}\sum_{\substack{r \in \L^*_+, v \in P_L\\ v\neq -r}} \big|\widehat{\o}_N(v+r) \big||\eta_r|\|b_{v}\xi\|\|(\cN_++1)^{1/2}\xi\|\\
		&\leq \frac{C}{N}\sum_{\substack{r \in \L^*_+, v \in P_L\\ v\neq -r}} \frac{\big|\widehat{\o}_N(v+r) \big|}{|r|^2} \|b_{v}\xi\|\|(\cN_++1)^{1/2}\xi\|\\
		& \leq C \frac{\log N}{N} \Big[\sum_{\substack{v \in P_L}} |v|^2 \|b_{v}\xi\|^2\Big]^{1/2}\Big[\sum_{\substack{v \in P_L}} \frac{1}{|v|^2} \Big]^{1/2}\|(\cN_++1)^{1/2}\xi\|\\
		&\leq C\frac{(\log N)^{3/2}}{N}\, \| \cK^{1/2}\xi\| \|(\cN_++1)^{1/2}\xi\|\, .
	\end{split}
	\]
The terms $Z_5$ and $Z_6$ can be bounded similarly. With Lemma 
\ref{lm:aprioriestimateKD}, we conclude that 
\begin{equation}
\label{eq:quartic-cQ}
	e^{-D}\,\wt{\cQ}_\cR e^D -\wt{\cQ}_\cR =  \sum_{i=1}^6  \int_0^1 ds\;  e^{-sD} Z_i e^{sD} = \widetilde{\cE}_{\cQ} 
	\end{equation}
where 
\[
\pm \widetilde{\cE}_{\cQ}  \leq C N^{-1/2} (\log N)^{1/2} \cK (\cN_++1)^{\k +3}\,.
\]
Combining (\ref{eq:DED}), (\ref{eq:LDKDL}), (\ref{eq:LDVDL}) and (\ref{eq:Ki}), we obtain (\ref{eq:expD}). 
\end{proof}


\subsection{Growth of the kinetic energy}
\label{sec:apriori}
In this section we show Lemma \ref{lm:aprioriestimateKD}, establishing a-priori bounds on the growth of the kinetic energy under the action of the unitary operator $e^D$. We will use the following preliminary estimate.
\begin{lemma}\label{lm:aprioriestimateKD-1} 
	Let $\cK$ be defined in \eqref{eq:cKVN} and $D$ defined as in \eqref{eq:defD}, with  $\a \geq 5/2$ and $\nu \in (0,1/2)$. Then for any $s \in [0,1]$,  there exists a constant $C>0$ such that  
		\begin{equation} \begin{split}		\label{eq:expDKexpD-1}
 e^{-sD}\cK e^{sD} & \leq \cK  + CN^{2(\a+\nu-1/2)} (\cN_++1)^{3}
					\end{split}
			\end{equation}
		for  all $N$ large enough.
\end{lemma}

\begin{proof}[Proof of Lemma \ref{lm:aprioriestimateKD-1}] The proof follows 
	from  Gronwall's lemma. For a fixed $\xi \in \cF_+^{\leq N}$ and $s\in [0; 1]$ we define
	\begin{equation*}
		h_\xi (s) := \langle \xi, e^{-sD} \cK\,e^{sD} \xi\rangle\,.  \end{equation*}
	Then 
	\begin{equation*}\label{eq:h'1} h'_\xi  (s) = \langle \xi, e^{-sD} [\cK\,, D] e^{sD} \xi\rangle. \end{equation*}
	With 
	\[\begin{split}
		[a^*_pa_p,a^*_{v+r}a^*_{w-r}a_va_{w}]&=a^*_{v+r} a^*_{w-r}a_v a_w (\d_{p,v+r} + \d_{p,w-r} - \d_{p,v} - \d_{p,w}) 
	\end{split}\]
	we find
	\[
	\begin{split}
		[\cK,D] 
		&= \fra 1{2N}\sum_{\substack{r \in \L^*_+, v,w \in P_L\\ v\neq -r, w\neq r}}(r^2 + 2 r \cdot v)\eta_r a^*_{v+r}a^*_{w-r}a_va_{w}+\hc
	\end{split}
	\]
	Using the scattering equation \eqref{eq:eta-scat} and the definition of $\widehat{\o}_N(p)$ in \eqref{eq:defomegaN} we get
	\begin{equation}
		\begin{split} \label{eq:commcKD-0}
			[\cK,D]=\; &
			-\fra 14\sum_{\substack{r \in \L^*, v,w \in P_L\\ v\neq -r, w\neq r}} (\widehat{V}(\cdot/e^N)*\widehat{f}_{N,\ell})(r) \big(a^*_{v+r}a^*_{w-r}a_va_{w}+\hc \big)\\ 
			& +\fra {1} {4 N} \sum_{\substack{r \in \L^*_+, v,w \in P_L\\ v\neq -r, w\neq r}}(\widehat{\o}_N*\widehat{f}_{N,\ell})(r) \big(a^*_{v+r}a^*_{w-r}a_va_{w}+\hc\big)\\
			& +\fra 1{4}\sum_{\substack{v,w \in P_L}} \big(\widehat{V}(\cdot/e^N)*\widehat{f}_{N,\ell}\big)(0)\big(a^*_{v}a^*_{w}a_va_{w}+\hc\big)\\
			&+ \frac 1 N \sum_{\substack{r \in \L^*, v,w \in P_L\\ v\neq -r, w\neq r}} 
			r \cdot v \,\eta_r \big( a^*_{v+r}a^*_{w-r}a_va_w +\hc \big)
			=: \; \sum_{i=1}^4 K_i\,.
		\end{split}
	\end{equation}
To estimate the first   term on the r.h.s of \eqref{eq:commcKD-0} we use (\ref{eq:intpotf}) to estimate $\|\widehat{V}(\cdot/e^N)*\widehat f_{N,\ell}\|_\infty \leq C /N$ and (\ref{eq:fell01}) to bound $\|\widehat{V}(\cdot/e^N)*\widehat f_{N,\ell}\|_2 \leq C e^N$. We obtain 
	\be \label{eq:intVfNell}
	\begin{split} & \sup_{v \in \L^*_+}\sum_{\substack{r \in \L^*\\ r\neq -v}} \frac{|(\widehat{V}(\cdot/e^N)*\widehat{f}_{N,\ell})(r)|}{|v+r|^2} \\
&\leq \;\|\widehat{V}(\cdot/e^N)*\widehat f_{N,\ell}\|_\infty \sum_{ |r+v|\leq e^N} \frac{1}{|v+r|^2} + \|\widehat{V}(\cdot/e^N)*\widehat f_{N,\ell}\|_2 \Big[ \sum_{ |r+v| > e^N} \frac{1}{|v+r|^4} \Big]^{1/2} \leq C\,.
	\end{split}\ee
Hence
	\[\begin{split}&|\langle \xi, e^{-sD}K_1e^{sD}\xi\rangle| \\
		&\quad \leq C \sum_{\substack{r \in \L^*, v,w \in P_L\\ v\neq -r, w\neq r}} |(\widehat{V}(\cdot/e^N)*\widehat{f}_{N,\ell})(r)|\|(\cN_++1)^{-1/2}a_{v+r}a_{w-r}e^{sD}\xi\|\|(\cN_++1)^{1/2}a_va_{w}e^{sD}\xi\|\\
		&\quad\leq C\Big[ \sum_{\substack{r \in \L^*, v,w \in P_L\\ v\neq -r, w\neq r}} \frac{|(\widehat{V}(\cdot/e^N)*\widehat{f}_{N,\ell})(r)|}{|v+r|^2} \|(\cN_++1)^{1/2}a_va_we^{sD}\xi\|^2\Big]^{1/2}\\
		&\quad\hskip1cm \times\Big[ \sum_{\substack{r \in \L^*, v,w \in P_L\\ v\neq -r, w\neq r}}|(\widehat{V}(\cdot/e^N)*\widehat{f}_{N,\ell})(r)| |v+r|^2 \|(\cN_++1)^{-1/2}a_{v+r}a_{w-r}e^{sD}\xi\|^2\Big]^{1/2}
		\\
		&\quad \leq CN^{\a+\n-1/2}\|(\cN_++1)^{3/2}e^{sD}\xi\|\|\cK^{1/2}e^{sD}\xi\|.
	\end{split}\]
Similarly, with $\|\widehat{\o}_N*\widehat f_{N,\ell}\|_\infty \leq \| \o_N\|_1 \leq C $ and $\|\widehat{\o}_N*\widehat f_{N,\ell}\|_2 \leq \| \o_N\|_2 \leq C N^\a $  we get
\[ \begin{split}
& \sup_{v \in \L^*_+}\sum_{\substack{r \in \L^* \\ r\neq -v}} \frac{|(\widehat{\o}_N*\widehat{f}_{N,\ell})(r)|}{|v+r|^2} \\
&\leq \;\|\widehat{\o}_N*\widehat f_{N,\ell}\|_\infty \sum_{ |r+v|\leq N^\a} \frac{1}{|v+r|^2} + \|\widehat{\o}_N*\widehat f_{N,\ell}\|_2 \Big[ \sum_{ |r+v|\geq N^\a} \frac{1}{|v+r|^4} \Big]^{1/2} \leq C \log N\,.
\end{split}\]
Hence
\[\begin{split}|\langle \xi, e^{-sD}K_{2}e^{sD}\;\xi \rangle | \; & \leq \frac CN\Big[\sum_{\substack{r \in \L^*_+, v,w \in P_L\\ v\neq -r, w\neq r}}\frac{|(\widehat\o_N \ast \widehat{f}_{N,\ell})(r)|}{|v+r|^2}\| a_va_w (\cN_++1)^{1/2}e^{sD} \xi\|^2\Big]^{1/2} \\ 
	&\;\hskip1cm\times
	\Big[\sum_{\substack{r \in \L^*_+, v,w \in P_L\\ v\neq -r, w\neq r}}|v+r|^2\|a_{v+r}a_{w-r}(\cN_++1)^{-1/2}e^{sD}\xi\|^2\Big]^{1/2}\\
	&\;\leq  C N^{\a+\n -1}(\log N)^{1/2} \|\cK^{1/2}e^{sD}\xi\| \|(\cN_++1)^{3/2}e^{sD}\xi\|\,.
\end{split}\]
As for $K_3$ we use Eq. \eqref{eq:intpotf} in Lemma \ref{lm:propomega} to conclude:
\be\label{eq:K3bound}\begin{split}
	\Big|\langle \xi,e^{-sD} K_3\; e^{sD}\xi\rangle \Big|& \leq \fra CN \sum_{\substack{v,w \in P_L}}\|a_va_we^{sD}\xi\|^2 \leq \fra CN \|(\cN_++1)e^{sD}\xi\|^2.
\end{split}\ee
Finally, to bound $K_4$ we write $r \cdot v = (r+v) \cdot v - |v|^2$ and we split correspondingly $K_4$ in two terms, denoted by $K_{41}$ and $K_{42}$ below. Recalling from Lemma \ref{lm:eta} that $\| \eta \|_2 \leq C N^{-\alpha}$, we bound
\[
\begin{split}
	|\langle\xi, e^{-sD}K_{41}\;e^{sD}\xi\rangle|& \leq \frac C N\, \Big(\sup_{v \in P_L}|v| \Big)\Big[\sum_{\substack{r \in \L^*,\,  v,w \in P_L\\ v\neq -r, w\neq r}}  | \eta_r |^2
	\|a_va_w (\cN_++1)^{1/2}e^{sD}\xi\|^2\Big]^{1/2}\\
& \hskip 1cm \times\Big[\sum_{\substack{r \in \L^*,\, v,w \in P_L\\ v\neq -r, w\neq r}} 
	\,  |r+v|^2 \,  \|a_{v+r}a_{w-r} (\cN_++1)^{-1/2}e^{sD}\xi\|^2\Big]^{1/2}\\
	&\; \leq C N^{\a+2\nu-1} \,\|\cK^{1/2}e^{sD}\xi\| \|(\cN_++1)^{3/2}e^{sD}\xi\|\,.
\end{split}
\]
On the other hand
\[
\begin{split}
	|\langle\xi, e^{-sD}K_{42}\;e^{sD}\xi\rangle|& \leq \frac C N\, \Big(\sup_{v \in P_L}|v| \Big)\Big[\sum_{\substack{r \in \L^*,\,  v,w \in P_L\\ v\neq -r, w\neq r}}  | \eta_r |^2 |v|^2
	\|a_va_w (\cN_++1)^{-1/2}e^{sD}\xi\|^2\Big]^{1/2}\\
& \hskip 1cm \times\Big[\sum_{\substack{r \in \L^*,\, v,w \in P_L\\ v\neq -r, w\neq r}} 
	\,  \|a_{v+r}a_{w-r} (\cN_++1)^{+1/2}e^{sD}\xi\|^2\Big]^{1/2}\\
	&\; \leq C N^{\a+2\nu-1} \,\|\cK^{1/2}e^{sD}\xi\| \|(\cN_++1)^{3/2}e^{sD}\xi\|\,.
\end{split}
\]
Collecting the results above (and recalling that $\n <1/2 $ and that $\cN_+$ commutes with $D$) we end up with
\begin{equation*} \label{eq:DKgrow}\begin{split}
\big| \langle \xi, e^{-sD} [\cK, D]  e^{sD} \xi \rangle \big|& \leq \langle \xi, e^{-sD} \cK e^{sD} \xi \rangle  + CN^{2 (\a+\n-1/2)} \langle \xi, (\cN_++1)^3 \xi \rangle\,.\\ 
\end{split}
\end{equation*}
Hence, applying Gronwall's lemma to the differential inequality 
\[
|h_\xi'(s)|\leq h_\xi(s) +CN^{2(\a+\n-1/2)} \langle \xi, (\cN_++1)^3 \xi \rangle 
\]
we end up with \eqref{eq:expDKexpD-1}. 
\end{proof}

With the help of  Lemma \ref{lm:aprioriestimateKD-1} we can now show Lemma  \ref{lm:aprioriestimateKD}.

\begin{proof}[Proof of Lemma \ref{lm:aprioriestimateKD}.] We first show that the commutator $[\cK,D]$ satisfies the bound
	\begin{equation}\label{eq:commKD}
		\pm [\cK,D] \leq C \frac{ (\log N)^{1/2}}{N^{1/2}}\,\cK \cN_+\,.
	\end{equation}
Indeed, the bounds for the terms $K_1$, $K_2$ and $K_4$ defined in \eqref{eq:commcKD-0} can be all improved by using the kinetic energy operator. We have  (recall the definition of $P_L$ in \eqref{eq:defPL})
		\begin{equation}\label{eq:K0}\begin{split}|\langle \xi, K_1\xi\rangle| 
			&\leq C\Big[ \sum_{\substack{r \in \L^*, v,w \in P_L:\\ r\neq -v, w}} \frac{|(\widehat{V}(\cdot/e^N)*\widehat{f}_{N,\ell})(r)|}{|v+r|^2}|v|^2 \|a_va_w\xi\|^2\Big]^{1/2}\\
			&\quad\hskip1cm \times\Big[ \sum_{\substack{r \in \L^*, v,w \in P_L:\\ r\neq -v, w}}|(\widehat{V}(\cdot/e^N)*\widehat{f}_{N,\ell})(r)| \,\frac{1}{|v|^2}\,|v+r|^2 \|a_{v+r}a_{w-r}\xi\|^2\Big]^{1/2}
			\\
			&\quad \leq CN^{-1/2}(\log N)^{1/2}\|\cK^{1/2}\cN_+^{1/2}\xi\|^2\,,
		\end{split}\end{equation}
and, similarly, 
\[\begin{split}|\langle \xi, K_{2}\;\xi \rangle | \; & \leq \frac CN\Big[\sum_{\substack{r \in \L^*_+, v,w \in P_L:\\ r\neq -v, w}}\frac{|(\widehat\o_N \ast \widehat{f}_{N,\ell})(r)|}{|v+r|^2}\, |v|^2\| a_va_w  \xi\|^2\Big]^{1/2} \\ 
	&\;\hskip1cm\times
	\Big[\sum_{\substack{r \in \L^*_+, v,w \in P_L:\\ r\neq -v, w}}\frac 1 {|v|^2}\,|v+r|^2\|a_{v+r}a_{w-r}\xi\|^2\Big]^{1/2}\\
	&\;\leq  C N^{-1} (\log N) \|\cK^{1/2} \cN_+^{1/2}\xi\|^2\,.
\end{split}\]
To show that $K_4$ is also bounded from the r.h.s. of \eqref{eq:commKD} we split as before $K_4$ into $K_{41}$ and $K_{42}$. We get  
\[
\begin{split}
	 |\langle\xi, K_{41}\;\xi\rangle|
& \leq \frac C N\, \Big[\sum_{\substack{r \in \L^*,\\  v,w \in P_L:\\ r\neq -v, w}}  | \eta_r |^2 |v|^2
	\|a_va_w \xi\|^2\Big]^{1/2} \Big[\sum_{\substack{r \in \L^*,\\ v,w \in P_L:\\ r\neq -v, w}} 
	\,  |r+v|^2 \,  \|a_{v+r}a_{w-r} \xi\|^2\Big]^{1/2}\\
	&\; \leq C N^{\nu-1} \,\|\cK^{1/2}\cN_+^{1/2}\xi\|^2\,.
\end{split}
\]
On the other hand, distinguishing the cases $r+v \in P_L$ and $r+v \in P_L^c$ we find
\[
\begin{split}
	&|\langle\xi, K_{42}\;\xi\rangle|\\
& \leq \frac C N\,\Big( \sup_{v \in P_L} |v|\Big)\Big[\sum_{\substack{r \in \L^*,\,  v,w \in P_L\\  w\neq r,\, r+v \in P_L}}  \frac{1}{|r+v|^2} |v|^2
	\|a_va_w \xi\|^2\Big]^{1/2}\\
& \hskip 2cm \times\Big[\sum_{\substack{r \in \L^*,\, v,w \in P_L\\ w\neq r, v+r \in P_L}} 
	\, | \eta_r |^2 |r+v|^2 \|a_{v+r}a_{w-r} \xi\|^2\Big]^{1/2}\\
& +\frac C N\,\Big[\sum_{\substack{r \in \L^*,\,  v,w \in P_L\\  w\neq r,\, r+v \in P_L^c}}  \frac{|v|^2}{|r+v|^2}\, |\eta_r |^2 \, |v|^2 
	\|a_va_w \xi\|^2\Big]^{1/2}\\
& \hskip 2cm \times\Big[\sum_{\substack{r \in \L^*,\, v,w \in P_L\\ w\neq r, v+r \in P_L^c}} \hskip -0.1cm
	  |r+v|^2 \|a_{v+r}a_{w-r} \xi\|^2\Big]^{1/2}\\
	&\; \leq C N^{\nu-1} (\log N)^{1/2} \,\|\cK^{1/2}\cN_+^{1/2}\xi\|^2 \,.
\end{split}
\]
Eq. \eqref{eq:commKD} then follows from the previous bounds, together with \eqref{eq:commcKD-0} and \eqref{eq:K3bound}.
  Now, with \eqref{eq:commKD} and using that $[\cN_+,D]=0$ we rewrite
\begin{equation}\label{eq:DKD-exp}\begin{split}
e^{-D} \cK e^{D}  =\; &\ \cK + \int_0^1 dt_1 \,e^{-t_1 D}[\cK,D]e^{t_1 D}\\
\leq &\; \cK + C\, \frac{(\log N)^{1/2}}{N^{1/2}} \int_0^1 dt_1 \,\cN_+^{1/2}e^{-t_1 D}\cK e^{t_1 D}\cN_+^{1/2}\,.
\end{split}\end{equation}
Iterating $\k-1$ times we obtain 
\[ \begin{split}
e^{-D} \cK e^{D}  \leq \; & 
 \cK + C \sum_{j=1}^{\k-1} \frac{(\log N)^{j/2}}{j! N^{j/2}}\,\cK\, \cN_+^{j} \\ &+C \frac{(\log N)^{\kappa /2}}{N^{\frac{\kappa}{2}}} \int_0^1 dt_1 \int_0^{t_1} dt_2 \dots \int_0^{t_{\k-1}} dt_\kappa \,  \cN_+^{\kappa/2} e^{-t_\kappa D} \cK e^{t_\kappa D} \cN_+^{\k/2}\,. \end{split} \]
 Estimating the error term with Lemma \ref{lm:aprioriestimateKD-1}, we find 
\[ \begin{split} 
e^{-D} \cK e^{D}  \leq \; & 
 \cK + C \sum_{j=1}^{\k} \frac{(\log N)^{j/2}}{j! N^{j/2}}\,\cK\, \cN_+^{j} +C \frac{(\log N)^{\kappa /2}}{\kappa! N^{\kappa/2}} N^{2(\alpha + \nu - 1/2)} (\cN_+ + 1)^{\kappa+3}\,. \end{split} \]
Choosing $\k > 4 \a + 4 \nu -2$ and $N$ large enough, we obtain (\ref{eq:DKD}). Applying (\ref{eq:commKD}) to the identity in (\ref{eq:DKD-exp}), we find
\[ \pm \Big[ e^{-D} \cK e^D - \cK \Big] = \pm \int_0^1 dt \, e^{-t D} \big[ \cK , D \big] e^{tD} \leq C \frac{(\log N)^{1/2}}{N^{1/2}} \int_0^1 dt \, \cN_+^{1/2} e^{-tD} \cK  e^{tD} \cN_+^{1/2}\,. \]
With (\ref{eq:DKD}), we arrive at (\ref{eq:DKD-K}). 
\end{proof}

		\subsection{Growth of the interaction potential $\cV_N$}

The aim of this section is to prove Lemma \ref{lm:growthVN}, describing the action of the quartic transformation $e^D$ on the potential energy operator $\cV_N$. To achieve this goal, we first prove estimates for the commutator $[\cV_N, D]$. 
\begin{lemma}
				\label{prop:commDcVN} Let $\cV_N$ and  $D$  be defined in \eqref{eq:cKVN} and \eqref{eq:defD} respectively, with $\a\geq 5/2$ and $\nu\in (0,1/2)$. Let $\cV_N^{(L)} $ be defined as in (\ref{eq:defVNL}). Then, for $N$ large enough, there exists a constant $C > 0$ such that  
				\begin{equation} \label{eq:commVN}
					[\cV_N,D]  =\; -\cV_N^{(L)} + \cE_{[\cV_N , D]},
				\end{equation}
				with
				\begin{equation}
					\label{eq:errorcommVN}
			\pm \cE_{[\cV_N , D]} \leq  CN^{\nu-1/2}(\log N)^{1/2} \cK \,(\cN_++1)^2\,.
				\end{equation}Moreover
				\begin{equation}
					\label{eq:commVNLowD}
					\pm [\cV_N^{(L)},D]  \leq CN^{\n-1/2} (\log N)^{1/2}\cK\, (\cN_++1)^2\,.
				\end{equation}
			\end{lemma}
		\begin{proof}
			First, we prove Eq. \eqref{eq:commVN},\eqref{eq:errorcommVN}. A straightforward computation leads us to 
			\be\label{eq:commquartico}\begin{split}
				& [a^*_{p+u}a^*_qa_pa_{q+u},\,a^*_{v+r}a^*_{w-r}a_va_{w}]\\& =\d_{q+u,v+r}a^*_{p+u}a^*_qa_pa^*_{w-r}a_va_w +\d_{q+u,w-r}a^*_{p+u}a^*_qa_pa^*_{v+r}a_va_w\\
				& \quad+\d_{p,v+r}a^*_{p+u}a^*_qa^*_{w-r}a_{q+u}a_va_w
				+\d_{p,w-r}a^*_{p+u}a^*_qa^*_{v+r}a_{q+u}a_va_w\\
				&\quad- \d_{q,v}a^*_{v+r}a^*_{w-r}a^*_{p+u}a_wa_pa_{q+u}-\d_{q,w}a^*_{v+r}a^*_{w-r}a^*_{p+u}a_va_pa_{q+u}\\
				&\quad -\d_{p+u,v}a^*_{v+r}a^*_{w-r}a_wa^*_qa_pa_{q+u}-\d_{p+u,w}a^*_{v+r}a^*_{w-r}a_va^*_qa_pa_{q+u}\,.
			\end{split}\ee
			Normal ordering the terms in the first and in the last lines we obtain:
			\begin{equation}\label{eq:commVND-0}
				\begin{split}
					&	[\cV_N,D] \\
					&=  \fra{1}{8N} \sum_{\substack{ p,q \in \L^*_+,\, u \in \L^*\\-u \neq p,q}}\hskip -0.2cm \widehat{V}(u/e^N) \hskip -0.2cm\sum_{\substack{r \in \L^*_+,\, v,w \in P_L\\ u\neq -v, w} } 
					\eta_r  \,[a^*_{p+u}a^*_qa_pa_{q+u},\,a^*_{v+r}a^*_{w-r}a_va_{w}-\hc]  \\
					&=\fra{1}{4N}\sum_{\substack{u\in \L^*,r \in \L^*_+\\ v,w \in P_L\\ u\neq -v, w} } 					\widehat{V}((u-r)/e^N)\,\eta_r \big(a^*_{v+u}a^*_{w-u}a_va_w + \hc\big) + \sum_{i=1}^3 V_i\,,
				\end{split}
			\end{equation}
	where
			\[
			\label{eq:commVND}
			\begin{split}
				V_1&= - \fra{1}{4N}\sum_{\substack{u \in \L^*, r \in \L^*_+, v,w \in P_L\\ v\neq -r, w\neq r} } 
				\widehat{V}(u/e^N)\eta_r \big(a^*_{v+r}a^*_{w-r}a_{w-u}a_{v+u} +\hc\big)\\
				V_2&=  \fra{1}{2N}\sum_{\substack{u \in \L^*, q, r \in \L^*_+, v,w \in P_L\\ u\neq +r-w, -q, v\neq -r, } } 
				\widehat{V}(u/e^N)\eta_r \big(a^*_{w-r+u}a^*_{q}a^*_{v+r}a_{q+u}a_va_w + \hc\big) \\
				V_3&=-\fra{1}{2N}\sum_{\substack{u \in \L^*, q, r \in \L^*_+, v,w \in P_L\\ v\neq -r, w\neq r,u \neq v,-q} } 
				\widehat{V}(u/e^N)\eta_r \big(a^*_{v+r}a^*_{w-r}a^*_{q}a_{w}a_{v-u}a_{q+u} +\hc\big)\,. 
			\end{split}
			\]
Using the definition $\eta_r=- N \widehat{w}_{N,\ell}(r)$, and $\widehat{w}_{N,\ell}(r)=\d_{r,0} -\widehat{f}_{N,\ell}(r)$ we further split the first term on the r.h.s. of \eqref{eq:commVND-0}, thus getting
		\begin{equation}\label{eq:commVND-02}
[\cV_N,D] =-\fra{1}{4}\sum_{\substack{u \in \L^*\\ v,w \in P_L\\ u\neq -v, w} } \widehat{V}(u/e^N)\, \big(a^*_{v+u}a^*_{w-u}a_va_w + \hc\big) + \sum_{i=1}^5 V_i 			\end{equation}
			with
			\[
			\label{eq:V4} 
\begin{split}
		V_4 & = \fra{1}{4}\sum_{\substack{u \in \L^*\\ v,w \in P_L\\ u\neq -v, w} } 				\widehat{V}(u/e^N)\eta_0 \big(a^*_{v+u}a^*_{w-u}a_va_w + \hc\big)\\
		V_5 &= \fra{1}{4}\sum_{\substack{u \in \L^*\\ v,w \in P_L\\ u\neq -v, w} } 
				(\widehat{V}(\cdot/e^N)\,*\widehat{f}_{N,\ell})(u) \big(a^*_{v+u}a^*_{w-u}a_va_w + \hc\big) \,.
			\end{split} \]
			To conclude the proof of (\ref{eq:commVN}), we are going to bound the terms $V_i$, $i=1,\ldots, 5$. We notice that $V_5=-K_1$ (see  \eqref{eq:commcKD-0}), hence it satisfies the bound in \eqref{eq:K0}. On the other hand, with $|\eta_0|\leq N^{-2\a}$ and 
\[ \label{eq:intV-r2}
\sup_{v \in \L^*_+} \sum_{u \in \L^*} \fra{|\widehat{V}(u/e^N)|}{|u+v|^2}  \leq C N
\]
(which can be proved similarly as in \eqref{eq:intVfNell}), with the difference that $\| \hat{V} (./e^N)\|_\infty \leq C$),  we have 
			\[
			\begin{split}
				|\langle \xi,V_4\xi\rangle|&\leq CN^{-2\a}\, \Big[\sum_{\substack{u \in \L^*, v,w \in P_L\\ u\neq -v, w} } \fra{|\widehat{V}(u/e^N)|}{|u+v|^2}  |v|^2\, \|a_{v}a_{w}\xi\|^2\Big]^{1/2}\\
				&\hspace{4cm}\times\Big[\sum_{\substack{u \in \L^*, v,w \in P_L\\ u\neq -v, w} } \frac{1}{|v|^2}\,|u+v|^2 \|a_{v+u}a_{w-u}\xi\|^2\Big]^{1/2}\\
				&\leq C N^{1/2 -2\a} (\log N)^{1/2}\|\cK^{1/2}(\cN_++1)^{1/2}\xi\|^2\,.
			\end{split}
			\]
			Next we bound $V_1$. We split $V_1$ in two terms $V_{11}$ and $V_{12}$, defined by restricting to the cases $v+r \in P_L$ and $v+r \in P_L^c$ respectively; we have
			\[
			\begin{split}
				|\langle \xi,V_{11}\xi\rangle|&\leq \fra CN \Big[\sum_{\substack{u \in \L^*, r \in \L^*_+, v,w \in P_L:\\ v+r \in P_L,\, v\neq -r, w\neq r} } \fra{|\widehat{V}(u/e^N)|}{|u+v|^2}\,|v+r|^2 \|a_{v+r}a_{w-r}\xi\|^2\Big]^{1/2}\\
				&\hspace{2cm}\times\Big[\sum_{\substack{u \in \L^*, r \in \L^*_+, v,w \in P_L:\\v+r \in P_L,\, v\neq -r, w\neq r} } \fra{|\eta_r|^2}{|v+r|^2}\,|u+v|^2 \|a_{v+u}a_{w-u}\xi\|^2 \Big]^{1/2}\\
				&\leq CN^{\nu-1/2}(\log N)^{1/2}\|\cK^{1/2}(\cN_++1)^{1/2}\xi\|^2\,.
			\end{split}
			\]
			As for $V_{12}$, we have
			\[
			\begin{split}
				|\langle \xi,V_{12}\xi\rangle|&\leq \fra CN \,\Big[\sum_{\substack{u \in \L^*, r \in \L^*_+, v,w \in P_L:\\ v+r \in P_L^c,\, v\neq -r, w\neq r} } \fra{|\widehat{V}(u/e^N)|}{|u+v|^2}|v+r|^2 \|a_{v+r}a_{w-r}\xi\|^2\Big]^{1/2}\sup_{v+r \in P_L^c}\fra{1}{|v+r|}\\
				&\hspace{2cm}\times\Big[\sum_{\substack{u \in \L^*, r \in \L^*_+, v,w \in P_L:\\v+r \in P_L^c,\, v\neq -r, w\neq r} } |u+v|^2 \|a_{v+u}a_{w-u}\xi\|^2|\eta_r|^2\Big]^{1/2}\\
				&\leq CN^{\nu-1/2}(\log N)^{1/2}\|\cK^{1/2}(\cN_++1)^{1/2}\xi\|^2.
			\end{split}
			\]
			Next, we focus on $V_2$. We get:
			\[
			\begin{split}
				|\langle\xi,V_2\xi\rangle| &\;\leq \frac CN \,
				\Big[\sum_{\substack{u \in \L^*, q, r \in \L^*_+, v,w \in P_L\\ u\neq +r-w, -q, v\neq -r, } } 
				\frac{|\widehat{V}(u/e^N)|^2}{|w-r+u|^2}\,|\eta_r|^2\,v^2\,\|a_va_{q+u}a_w\xi\|^2\Big]^{1/2}\\
				&\hspace{2cm}\times \Big[\sum_{\substack{u \in \L^*, q, r \in \L^*_+, v,w \in P_L\\ u\neq +r-w, -q, v\neq -r, } } \frac{|w-r+u|^2}{|v|^2} \|a_{w-r+u}a_{q}a_{v+r}\xi\|^2 \Big]^{1/2}\\
				&\;\leq CN^{\n-1/2}(\log N)^{1/2}\|\cK^{1/2}(\cN_++1)\xi\|^2\,.
			\end{split}
			\]
			Finally to estimate $V_3$, we consider the contributions coming from $q \in P_L$ and $q \in P_L^c$ separately, which we denote with $V_{31}$ and $V_{32}$ respectively. We get
			\[
			\begin{split}
				|\langle\xi,V_{31}\xi\rangle| &\;\leq  \fra CN \Big[\sum_{\substack{u \in \L^*, r \in \L^*_+, q,v,w \in P_L\\ v\neq -r, w\neq r,u \neq v,-q} } 
				\frac{|\widehat{V}(u/e^N)|^2}{|v-u|^2}\,|q|^2\, \|a_{q}a_{v+r}a_{w-r}\xi\|^2 \Big]^{1/2}\\
				&\hspace{2cm}\times\Big[\sum_{\substack{u \in \L^*, r \in \L^*_+, q,v,w \in P_L\\ v\neq -r, w\neq r,u \neq v,-q} } |\eta_r|^2\frac{1}{|q|^2}|v-u|^2 \|a_{w}a_{v-u}a_{q+u}\xi\|^2\Big]^{1/2}\\
				& \;\leq CN^{\nu-1/2}(\log N)^{1/2}\|\cK^{1/2}(\cN_++1)\xi\|^2\,,
			\end{split}
			\]
			and 
			\[
			\begin{split}
				|\langle\xi,V_{32}\xi\rangle| &\;\leq  \fra CN\,\Big[\sum_{\substack{u \in \L^*, r \in \L^*_+, \\v,w \in P_L, q\in P_L^c\\ v\neq -r, w\neq r,u \neq v,-q} } 
				\frac{|\widehat{V}(u/e^N)|^2}{|v-u|^2}\,|q|^2\, \|a_{q}a_{v+r}a_{w-r}\xi\|^2 \Big]^{1/2}\sup_{q \in P_L^c}\frac{1}{|q|}\\
				&\hspace{2cm}\times\Big[\sum_{\substack{u \in \L^*, r \in \L^*_+,\\ v,w \in P_L, q\in P_L^c\\ v\neq -r, w\neq r,u \neq v,-q} } |\eta_r|^2|v-u|^2 \|a_{w}a_{v-u}a_{q+u}\xi\|^2\Big]^{1/2}\\
				& \;\leq CN^{\nu-1/2} \|\cK^{1/2}(\cN_++1)\xi\|^2\,.
			\end{split}
			\]
			This concludes the proof of (\ref{eq:commVN}), (\ref{eq:errorcommVN}). In order to show (\ref{eq:commVNLowD}), we observe that 
\begin{equation}\label{eq:VNLres} \cV_N^{(L)} = \frac{1}{4N} \sum_{u,v,w \in \L^*} \hat{V} (u/e^N) a_{v+u}^* a_{w-u}^* a_v a_w \left[ \chi (v,w \in P_L) + \chi (v+u , w-u \in P_L) \right]\,.  \end{equation} 
A part from the restrictions on the momenta, this is just the potential energy operator $\cV_N$. Thus, the commutator $[\cV_N^{(L)}, D]$ will produce the same terms as the commutator $[\cV_N , D]$, just with additional restrictions on the momenta. We already proved that the operators $V_1,V_2, V_3$ on the r.h.s. of  (\ref{eq:commVND-0}) can be bounded by the r.h.s. of (\ref{eq:errorcommVN}); this will not change with the additional constraints. To conclude the proof of  (\ref{eq:commVNLowD}), we only have to show that also the first sum on the r.h.s. of (\ref{eq:commVND-0}), when restricted to momenta determined by (\ref{eq:VNLres}), can be bounded by the r.h.s. of 
(\ref{eq:errorcommVN}). This follows from 
\[ \begin{split}  \Big| \frac{1}{N} \sum_{\substack{u \in \L^*, r \in \L^*_+, v,w \in P_L : \\ v+r, w-r \in P_L}}&\widehat{V} ((u-r)/e^N) \eta_r \langle \xi, a_{v+u}^* a_{w-u}^* a_v a_w \xi \rangle \Big| \\ \leq \; &\frac{1}{N} \sum_{\substack{u \in \L^*, r \in \L^*_+, v,w \in P_L : \\ v+r, w-r \in P_L}} |\widehat{V} ((u-r)/e^N)| \, |\eta_r| \| a_{w-u} a_{v+u} \xi \| \|   a_v a_w \xi \| \\ \leq \; &\frac{1}{N} \Big[  \sum_{\substack{u \in \L^*, r \in \L^*_+, v,w \in P_L : \\ v+r, w-r \in P_L}} \frac{|\eta_r|^2}{v^2} (w-u)^2 \| a_{w-u} a_{v+u} \xi \|^2 \Big]^{1/2}  \\ &\hspace{2cm} \times  \Big[   \sum_{\substack{u \in \L^*, r \in \L^*_+, v,w \in P_L : \\ v+r, w-r \in P_L}} \frac{|\widehat{V} ((u-r)/e^N)|^2}{(w-u)^2}  v^2 \|   a_v a_w \xi \|^2 \Big]^{1/2}  \\ \leq \; &C N^{\nu-1/2} (\log N)^{1/2} \| \cK^{1/2} (\cN_+ + 1)^{1/2} \xi \|^2\,. 
\end{split} \] 
\end{proof}
		
With Lemma \ref{prop:commDcVN}, we can now show the validity of Lemma \ref{lm:growthVN}.

		\begin{proof}[Proof of Lemma \ref{lm:growthVN}]  We write  
			\[
			\begin{split}e^{-D}\cV_N e^{D} &=\; \cV_N + \int_0^1ds\; e^{-sD}[\cV_N,D]e^{sD}\\
				& =\; \cV_N +\int_0^1ds\;e^{-sD}[ -\cV_N^{(L)} +\cE_{[\cV_N , D]}]e^{sD}\,.
				\end{split}\]
Expanding once more the integral and using \eqref{eq:errorcommVN},\eqref{eq:commVNLowD} as well as Lemma \ref{lm:aprioriestimateKD}, we obtain
			\be\begin{split}  \label{eq:expD-cV-expD}
e^{-D}\cV_N e^{D} & = \;\cV_N -\cV_N^{(L)}+ \int_0^1 ds\,  e^{-sD} \cE_{[\cV_N , D]} e^{sD} + \int_0^1ds\int_0^s dt \,e^{-tD}[\cV_N^{(L)}  ,D]e^{tD} \\ 
			&=: \cV_N^{(H)} + \cE_{\cV_N}  \\
\end{split}\ee
with $\cV_N^{(H)}$ as defined in (\ref{eq:defcV-H}) and where 
\be \label{eq:tildecEVN}
\pm \cE_{\cV_N}  \leq  CN^{\n-1/2}(\log N)^{1/2}  \cK(\cN_++1)^{\k+4}\,.
\ee
Here $\k \in \bN$ is the smallest integer such that $\k> 4(\a+\n -1/2)$ and $N$ is large enough.
\end{proof}
			

\section{Diagonalization of quadratic Hamiltonians}
\label{sec:diag}

From Prop.\,\ref{prop:expD-cRN}, we observe that the renormalized Hamiltonian $e^{-D}\cR_N e^D$ can be approximated, on low-momentum states and up to negligible errors, by the quadratic (Bogoliubov) Hamiltonian 
\begin{equation}\label{eq:RNBog} \cR_N^\text{Bog} = C_\cR + Q_\cR^{(L)} \end{equation} 
which we are going to diagonalize in Sect.\,\ref{sub:dia-up}; this will be used later to prove upper bounds on the eigenvalues of (\ref{eq:Ham0}). 

On the other hand, by Prop. \ref{prop:RN}, the excitation Hamiltonian $\cR_N$ can be bounded below by the quadratic operator appearing on the r.h.s. of (\ref{eq:prRN-3}), which will be diagonalized in Sect.\,\ref{sub:dia-low}; this will allow us later to establish lower bounds on the eigenvalues of (\ref{eq:Ham0}).

\subsection{Diagonalization of (\ref{eq:RNBog})} 
\label{sub:dia-up} 

For $p \in \L^*_+$ we introduce the notation 
\begin{equation}
\label{eq:defFpGp}
\text{F}_p = p^2+\widehat{\o}_N(p), \qquad \text{G}_p =  \widehat{\o}_N(p).
\end{equation}
\begin{lemma}
	\label{lm:boundsFG}
	Let $V \in L^3(\bR^2)$ be non-negative, compactly supported and spherically symmetric. Let $\text{F}_p$ and $\text{G}_p$ be defined as in \eqref{eq:defFpGp}. Then there exists a constant $C>0$ such that
	\[(i)\;\fra{p^2}2\leq \text{F}_p\leq C(1+p^2), \qquad (ii)\;|\text{G}_p|\leq \fra{C}{(1+|p|/N^\a)^{3/2}}, \qquad(iii)\;|\text{G}_p|<\text{F}_p\]
	for all $p\in \L^*_+$. 
\end{lemma}
\begin{proof} 
The upper bound in i) follows easily from Lemma \ref{lm:eta}. For the lower bound we use that, from Lemma \ref{lm:eta}, $\widehat{\o}_N(p)\geq 0$ for $|p| \leq N^\a$ and $|\widehat{\o}_N(p)|\leq C N^{3\a/2}/|p|^{3/2} < p^2/2$, for $|p| \geq N^\alpha$.

Part ii) follows from Lemma \ref{lm:eta}. Finally, we show iii). On the one hand we have $\text F_p - \text{G}_p = p^2 >0$; on the other hand it is easy to show that  $\text F_p + \text G_p = p^2 + 2\, \widehat \o_N(p) \geq p^2/2 >0$, arguing as we did for the lower bound in part i). Thus, $|\text G_p| < \text{F}_p$.
\end{proof}
By Lemma \ref{lm:boundsFG}, part (iii), we can introduce, for an arbitrary $p \in \Lambda^*$, the coefficient $\tau_p$, requiring that 
\begin{equation*}
\label{eq:deftau}
\tanh(2\tau_p)= - \fra{\text G_p}{\text{F}_p }\,.
\end{equation*}
We define the antisymmetric operator 
\begin{equation}\label{eq:defBt}B_\t = \fra 12 \sum_{p\in P_L}\t_p(b^*_{-p}b^*_{p}-b_{-p}b_p)\,,\end{equation}
with the low-momentum set $P_L$ defined in \eqref{eq:defPL}. The generalized Bogoliubov transformation $e^{B_\tau}$ has the following properties.
\begin{lemma}
	\label{lm:growcNcH}
Let $B_\t$ be defined in \eqref{eq:defBt}. Then, under the same assumptions of Theorem  \ref{thm:main} and for any $k \in \bN$, there exists a constant $C_k>0$ (depending on $k$) such that 
\begin{equation}
	\label{eq:growthcNcHD}
	\begin{split}
	e^{-B_\t}(\cN_++1)^ke^{B_\t}&\leq C_k(\cN_++1)^k\\
	e^{-B_\t} (\cK+1) (\cN_++1)^ke^{B_\t}&\leq C_k \cK (\cN_++1)^k+ C_k(\log N)(\cN_++1)^{k+1}\\
	e^{-B_\t}\cV_N(\cN_++1)^ke^{B_\t}&\leq C_k\cV_N(\cN_++1)^k+ C_k (\log N)^2(\cN_++1)^{k+2} \,.
	\end{split}
	\end{equation}
	%
\end{lemma}
\begin{proof}
We proceed similarly as in \cite[Lemma 5.2]{BBCS4}. From Lemma \ref{lm:boundsFG} and from $|\tau_p| \leq C |G_p|/ F_p$, we easily obtain 
\begin{equation}\label{eq:bds-tau} \| \tau \|_2 \leq C , \qquad \| \tau \|^2_{H^1} , \| \tau \|_1 \leq C \log N\,. \end{equation} 
To show the first bound in (\ref{eq:growthcNcHD}), for $k=1$, we consider, for a fixed $\xi \in  \cF_+^{\leq N}$, 
\[  f_\xi(s) = \langle\xi, e^{-sB_\t} (\cN_+ + 1) e^{sB_\t}\xi\rangle\,. \] 
With 
\[ f'_\xi(s) = \langle\xi, e^{-sB_\t}[(\cN_++1),B_\t]e^{sB_\t}\xi\rangle =  \sum_{p \in P_L} \tau_p \langle \xi, e^{-sB_\t} (b_p b_{-p} + b_p^* b_{-p}^*) e^{sB_\t} \xi \rangle \,,\]
and using $\| \tau \|_2 \leq C$, we obtain $|f'_\xi (s)| \leq C f_\xi (s)$. With Gronwall, we obtain  the first bound in (\ref{eq:growthcNcHD}), for $k=1$. The case $k > 1$ can be handled similarly. 

As for the second estimate in (\ref{eq:growthcNcHD}), let us consider the case $k=0$. For $\xi \in \cF_+^{\leq N}$, we set 
\[ g_\xi (s) = \langle\xi, e^{-sB_\t} \cK e^{sB_\t} \xi \rangle \]
and we compute 
\[ g'_\xi (s) = \langle \xi, e^{-sB_\t}[\cK,\,B_\t ] e^{sB_\t} \xi \rangle = \sum_{p \in P_L} p^2 \tau_p \langle \xi ,  e^{-sB_\t} (b_p b_{-p} + b^*_p b_{-p}^*)  e^{s B_\t} \xi \rangle\,. \]
Using $\| \tau \|^2_{H^1} \leq C \log N$, we obtain 
\[ \begin{split}  |g'_\xi (s)| &\leq C (\log N)^{1/2} \| \cK^{1/2} e^{s B_\t} \xi \| \| (\cN_+ +1)^{1/2} e^{s B_\t} \xi \| \\ &\leq g_\xi (s) + C (\log N) \langle \xi, e^{-s B_\t} (\cN_+ + 1) e^{s B_\t} \xi \rangle \leq g_\xi (s) + C (\log N) \langle \xi, (\cN_+ + 1) \xi \rangle\,, \end{split} \] 
where we used the estimate for the growth of $\cN_+$, shown above. By Gronwall, we obtain the second bound in (\ref{eq:growthcNcHD}), for $k=0$. The case $k > 0$ can be treated analogously (in this case, $g'_\xi (s)$ contains an additional contribution, arising from the commutator of $B_\t$ with $(\cN_+ + 1)^k$, which can also be treated similarly; for more details, see \cite[Lemma 5.2]{BBCS4}). 

Finally, let us show the last estimate in (\ref{eq:growthcNcHD}), focussing again on the case $k=0$. For fixed $\xi \in \cF_+^{\leq N}$, we define 
\[ h_\xi(s) = \langle\xi, e^{-sB_\t}\cV_N e^{sB_\t} \xi\rangle \, . \]
We have 
\[ \begin{split} 
h_\xi'(s) &= \langle\xi, e^{-sB_\t} [\cV_N , B_\t] e^{sB_\t}\xi\rangle \\ 
= \; &\fra 12 \sum_{w \in P_L, r\in \L^*_+}\widehat{V}(r/e^N)\t_w  \langle\xi, e^{-sB_\t} (b^*_{w-r}b^*_{-w+r}+b_{w-r}b_{-w+r}) e^{sB_\t}\xi\rangle \\ &+ \sum_{\substack{v \in P_L,r\in \L^*\\w\in \L^*_+}}\widehat{V}(r/e^N)\t_v  \langle\xi, e^{-sB_\t} (b^*_{v+r}b^*_{-v}a^*_{w-r}a_w+\hc) e^{sB_\t}\xi\rangle\,. \end{split} \] 
Switching to position space, we find 
\[ \begin{split} h_\xi'(s) = &\; \fra 12 \int_{\L^2}dxdy\;  \sum_{w\in P_L}e^{-iw \cdot (x-y)}\t_w \; e^{2N}V(e^N(x-y)) \langle\xi, e^{-sB_\t} (\check b^*_x\check b^*_y+\check b_x\check b_y)e^{sB_\t}\xi\rangle  \\
	&+ \int_{\L^2} dxdy\;  \sum_{v\in P_L}e^{-iv\cdot x}\t_v \; e^{2N}V(e^N(x-y)) \langle\xi, e^{-sB_\t} (\check b^*_x b^*_{-v}\check{a}^*_{y}\check a_y+\hc) e^{sB_\t}\xi\rangle \, .
	\end{split} \]
Using $\|\t\|_1\leq C\log N$, $\|\t\|\leq C$ and the first estimate in (\ref{eq:growthcNcHD}) (for $k=2$), we find 
\[  \begin{split} 
|h'_\xi (s)| &\leq C (\log N) \| \cV_N^{1/2} e^{-s B_\t} \xi \| \| \xi \| + C  \| \cV_N^{1/2} e^{s B_\t} \xi \| \| (\cN_+ + 1) e^{s B_\t} \xi \| \\ &\leq h_\xi (s)  + C (\log N)^2  \langle \xi , (\cN_+ + 1)^2 \xi \rangle\,. \end{split} \]
By Gronwall, we obtain the last bound in (\ref{eq:growthcNcHD}), for $k=0$. The case $k > 0$ can be treated similarly. 
\end{proof}

In the next proposition, we show that conjugation with the generalized Bogoliubov transformation $e^{B_\tau}$ diagonalizes the quadratic Hamiltonian (\ref{eq:RNBog}), up to negligible errors. 
\begin{prop} \label{prop:gsandspectrum}
	Let $V \in L^3(\bR^2)$ be non-negative, compactly supported and spherically symmetric. Let $\cR_N^\text{Bog}$ be defined as in (\ref{eq:RNBog}) (with $C_\cR$ and $Q_\cR^{(L)}$ defined in \eqref{eq:defC-cR} and, respectively, \eqref{eq:defQNL} with parameters $\alpha \geq 5/2$ and $\nu  \in (0;1/2)$). Let
\be \label{eq:EBog}
S_{\rm{Bog}}  := \fra 12 \sum_{p\in \L^*_+} \Big(\sqrt{p^4+8\pi p^2}-p^2-4\pi + \fra{(4\pi)^2}{2p^2}\Big)\,.
\ee
Then
\be \label{eq:cM}
 e^{-B_\t}\cR_N^{\rm{Bog}} e^{B_\t} = E_N^{\rm{Bog}} + \sum_{p\in P_L}\sqrt{p^4+8\pi p^2}\;a_p^*a_p  + \d_{\rm{Bog}}
\ee
where 
\be		\label{eq:EcM}
			E_N^{\rm{Bog}}= 2 \pi (N-1) + \pi^2 \frak{a}^2 +S_{\rm{Bog}}-4\pi^2  \sum_{\substack{p\in \L^*_+}}\fra{J_0(|p| \frak{a})}{|p|^2}\,,
			\ee
and the error term $\d_{\rm{Bog}}$ is bounded by
\[ \pm \d_{\rm{Bog}}  \leq  C N^{-1/2}(\log N)(\cH_N+1)(\cN_++1)+ CN^{-3\nu} \,,\]
for $N$ large enough.
\end{prop}

\begin{proof}
Proceeding very similarly as in \cite[Lemma 5.3]{BBCS4}, using the bounds (\ref{eq:bds-tau}), we obtain 
\be \begin{split} \label{eq:lm3.4-1}
e^{-B_\t} \cR_N^{\text{Bog}} e^{B_\t} =\;&  C_{\cR} +   \fra 12\sum_{p\in P_L}\big(\sqrt{F_p^2- G_p^2}-F_p\big)+\sum_{p\in P_L}\sqrt{ F_p^2-G_p^2} \, a_p^*a_p+\d_1
\end{split}\ee 
with $C_\cR$, $F_p, G_p$ as defined in \eqref{eq:defC-cR} and, respectively,  (\ref{eq:defFpGp}), and where 
\begin{equation*}
\label{eq:errorQ}
\pm \delta_1 \leq \fra C N (\log N)^2 (\cH_N+1)(\cN_++1)\,.
\end{equation*}
We have $F_p^2 - G_p^2 = (F_p - G_p) (F_p + G_p) = |p|^4 + 2 p^2 \widehat{\omega}_N (p)$. With the estimate (see the definition (\ref{eq:defomegaN}), recall from Lemma \ref{lm:eta} that 
$|g_N| \leq C$ and use the continuity of $\widehat{\chi}$ at the origin) 
\[ |\widehat{\o}_N(p)-\widehat{\o}_N(0)|\leq C|p|N^{-\a} \]
we can bound 
\[ \begin{split}
\Big|\sum_{p\in P_L}\Big[\sqrt{|p|^4+2\widehat{\o}_N(p) p^2} &-\sqrt{|p|^4+2\widehat{\o}_N(0) p^2}\Big]\; \langle \xi , a_p^*a_p \xi \rangle \Big|\\ & \leq CN^{-\a}\sum_{p\in P_L}|p|\langle \xi, a^*_p a_p \xi \rangle \leq CN^{-\a} \langle \xi , \cK \xi \rangle \,.
\end{split} \]
With $|\sqrt{p^4+2\widehat{\o}_N(0)p^2}-\sqrt{p^4+8\pi p^2}| \leq C|\widehat{\o}_N(0)-4\pi|\leq  C (\log N)/N$ (see (\ref{eq:omegahat0}) and (\ref{eq:intpotf})), we conclude that 
\begin{equation} \label{eq:BogSpectrum}
	\begin{split}
	\sum_{p\in P_L}\sqrt{ F_p^2-G_p^2} \; a_p^*a_p = \sum_{p\in P_L}\sqrt{p^4+8\pi p^2}\; a_p^*a_p +\d_2
	\end{split}
	\end{equation}
	where $\pm\d_2  \leq C N^{-1} (\log N) (\cK+1)$, for all $\alpha \geq 1$.

Let us now consider the constant term on the r.h.s. of \eqref{eq:lm3.4-1}. From (\ref{eq:defC-cR}) and (\ref{eq:defFpGp}), we obtain (adding and subtracting the factor $\sum_{p\in P_L} \widehat{\o}_N^2(p) / (4p^2)$)  
\be \label{eq:lm3.4-2}
	\begin{split}
	&	  C_{\cR} +   \fra 12\sum_{p\in P_L}\big(\sqrt{F_p^2- G_p^2}-F_p\big) \\
&=\frac N2 \big(\widehat{V}(\cdot/e^N)*\widehat{f}_{N,\ell}\big)(0) (N-1)  + \frac 12\sum_{p\in \L^*_+} \widehat{\o}_N(p)\eta_p - \fra 14 \sum_{p\in P_L }\fra{\widehat{\o}_N^2(p)}{p^2} \\
&\hskip 2cm +\fra 12\sum_{p\in P_L} \bigg(-p^2 -\widehat{\o}_N(p) +\sqrt{p^4+2\widehat\o_N(p)p^2 } +\frac 12  \fra{\widehat{\o}_N^2(p)}{p^2} \bigg)\,.
	\end{split}\ee
Expanding the square root, we find 
\begin{equation}\label{eq:bdB-4} \Big| -p^2 -\widehat{\o}_N(p) +\sqrt{p^4+2\widehat\o_N(p)p^2 } +\frac 12  \fra{\widehat{\o}_N^2(p)}{p^2}  \Big| \leq C |p|^{-4} \end{equation} 
uniformly in $N$. Up to an error vanishing as $N^{-1/2}$, we can therefore restrict the sum on the last line of (\ref{eq:lm3.4-2}) to $|p| < N^{1/4}$. After this restriction, we can use $| \widehat{\omega}_N (p) - 4\pi | \leq C |p| N^{-\alpha} + C(\log N) / N$, to replace $\widehat{\omega}_N (p)$ by $4\pi$. Comparing with (\ref{eq:EBog}) (and noticing that (\ref{eq:bdB-4}) remains true, if we replace $\widehat{\omega}_N (p)$ with $4\pi$), we conclude that 
\be \label{eq:SBog}
 \Big| \fra 12\sum_{p\in P_L} \bigg(-p^2 -\widehat{\o}_N(p) +\sqrt{p^4+2\widehat\o_N(p)p^2 } +\frac 12  \fra{\widehat{\o}_N^2(p)}{p^2} \bigg) - S_\text{Bog} \Big| \leq C \frac{\log N}{\sqrt{N}}\,. \ee
Let us now consider the terms on the second line of (\ref{eq:lm3.4-2}). First of all, we observe that, by \eqref{eq:intpotf},  
\[
\frac N2 \big(\widehat{V}(\cdot/e^N)*\widehat{f}_{N,\ell}\big)(0) (N-1) = 2 \pi (N-1) - 2\pi \big(\log(\ell/\aa) -\tfrac 12\big)    + \cO(\log N/N)\,.
\]
As for the second term on the r.h.s. of (\ref{eq:lm3.4-2}), we use the scattering equation (\ref{eq:eta-scat}) and the definition \eqref{eq:defomegaN} to write 
\begin{equation}\label{eq:etaterm}\begin{split}
		\fra 12\sum_{p\in \L^*_+} \widehat{\o}_N(p)\eta_p&= -\fra N4\sum_{p\in \L^*_+} \widehat{\o}_N(p)\fra{(\widehat{V}(\cdot/e^N)*\widehat{f}_{N,\ell})(p)}{p^2} + \fra{1}{4}\sum_{p\in \L^*_+}\fra{\widehat\o^2_N(p)}{p^2}\\
		&\hspace{0.5cm}+ \fra 1{4N}\sum_{p\in \L^*_+}\widehat{\o}_N(p)\fra{(\widehat{\o}_N*\eta)(p)}{p^2}\,.\end{split}\end{equation}
Since $\| \widehat{\omega}_N * \eta \|_\infty \leq \| \widehat{\omega}_N \| \| \eta \| \leq C$, the last term on the r.h.s. of (\ref{eq:etaterm}) is negligible, of order $(\log N)/N$. The second term on the r.h.s. of (\ref{eq:etaterm}), on the other hand, cancels with the third term on the r.h.s. of (\ref{eq:lm3.4-2}), up to a small error of order $N^{-3\nu}$ (because $\sum_{p \in P_L^c} \widehat{\omega}_N^2 (p) / p^2 \leq N^{-3\nu}$, from Lemma \ref{lm:eta} and by the definition (\ref{eq:defPL}) of the set $P_L$). Finally, to estimate the first term on the r.h.s. of (\ref{eq:etaterm}), we use that 
\[ \Big|(\widehat{V}(\cdot/e^N)*\widehat{f}_{N,\ell})(p)- (\widehat{V}(\cdot/e^N)*\widehat{f}_{N,\ell})(0) \Big| \leq C e^{-N} |p|\,,  \]
that $|(\widehat{V}(\cdot/e^N)*\widehat{f}_{N,\ell})(0)- 4 \pi /N | \leq C N^{-2} \log N$ and that, with $g_N$ as defined in (\ref{eq:defomegaN}), $|g_N - 4| \leq C /N$ by (\ref{eq:omegahat0}). We arrive at 
\[ C_{\cR} +   \fra 12\sum_{p\in P_L}\big(\sqrt{F_p^2- G_p^2}-F_p\big) =  2 \pi (N-1) - 2\pi \big(\log(\ell/\aa) -\tfrac 12\big)  + S_\text{Bog} - 4\pi \sum_{p \in \L^*_+} \frac{\widehat{\chi} (\ell p)}{p^2} + \delta_3 \] 
where $\pm \delta_3 \leq C (\log N)/\sqrt{N}$. Let us now compute the remaining sum. To this end, we observe that, denoting by $J_n$ the Bessel function of order $n$, we have 
\be \label{eq:x2chiell}
 \widehat{|.|^2 \chi_\ell} (p) = -8 \pi \ell \left[ \frac{J_1 (\ell p)}{|p|^3} - \frac{\ell}{2} \frac{J_0 (\ell p)}{|p|^2} - \frac{\ell^2}{4} \frac{J_1 (\ell p)}{|p|} \right]\,, \ee
which can be proved similarly to  \eqref{eq:chiellp}.
Hence, with \eqref{eq:chiellp} and \eqref{eq:x2chiell} we find 
 \[ \frac{\widehat{\chi} (\ell p)}{p^2} = - \frac{1}{4\ell^2} \widehat{|.|^2 \chi_\ell} (p) + \pi \frac{J_0 (\ell p)}{p^2} + \frac{\ell^2}{4} \widehat{\chi} (\ell p) \]
 and thus 
 \[\begin{split}  -4\pi \sum_{p \in \L^*_+} \frac{\widehat{\chi} (\ell p)}{p^2} &= -\frac{\pi}{\ell^2}  \widehat{|.|^2 \chi_\ell} (0) - 4 \pi^2 \sum_{p \in \L^*_+} \frac{J_0 (\ell p)}{p^2} - \pi \left[ \chi_\ell (0) - \widehat{\chi}_\ell (0) \right] \\ &= - 4 \pi^2  \sum_{p \in \L^*_+} \frac{J_0 (\ell p)}{p^2} - \pi + \frac{\pi^2 \ell^2}{2} \,, \end{split} \]
where we used that $\widehat{\chi}_\ell (0) = \pi \ell^2$ and $\widehat{|.|^2 \chi_\ell} (0) = \pi \ell^4 /2$. Taking into account that $\ell^2 \simeq N^{-2\alpha}$ and $\alpha \geq 5/2$, we conclude that 
\begin{equation}\label{eq:const-last}  C_{\cR} +   \fra 12\sum_{p\in P_L}\big(\sqrt{F_p^2- G_p^2}-F_p\big) = 2 \pi (N-1)   + S_\text{Bog} + I_\ell  + \delta_3 \end{equation} 
where $\pm \delta_3 \leq C (\log N)/ \sqrt{N}$ and where we defined 
\begin{equation}\label{eq:Iell}  I_\ell = -2\pi \log (\ell / \frak{a}) + \pi^2 \ell^2 - 4\pi^2 \sum_{p \in \L^*_+} \frac{J_0 (\ell p)}{p^2}\,. \end{equation} 
We claim now that the value of $I_\ell$ is independent of the choice of $\ell$ (this is why it is convenient to include the factor $\pi^2 \ell^2$ in the definition of $I_\ell$, despite the fact that this term is very small, for $\ell = N^{-\alpha}$, $\alpha \geq 5/2$). In fact, with the identity  
\[ (\widehat{\log(|\cdot|/\ell)\chi_\ell})(p) = 2\pi\Big[-\fra1{|p|^2} +\fra{J_0(\ell|p|)}{|p|^2}\Big],
\]
and using that $J_0 (z) = 1 - (z/2)^2 + \cO (z^4)$ close to $z=0$, we find that, for any $\ell_1, \ell_2 > 0$, 
\[  -4\pi^2 \sum_{p \in \L_+^*} \frac{J_0 (\ell_1 |p|)}{p^2} + 4 \pi^2 \sum_{p \in \L_+^*} \frac{J_0 (\ell_2 |p|)}{p^2}
= 2\pi \log (\ell_1 / \ell_2) - \pi^2 (\ell_1^2 - \ell_2^2)\,. \]
which implies that $I_{\ell_1} - I_{\ell_2} = 0$. Since $I_\ell$ is independent of $\ell$, we can evaluate the r.h.s. of (\ref{eq:const-last}) choosing for example $\ell = \frak{a}$. This completes the proof of (\ref{eq:EcM}). 
\end{proof}

\subsection{Diagonalization of quadratic Hamiltonian for lower bounds} 
\label{sub:dia-low}


Next, we discuss how to diagonalize the quadratic operator on the r.h.s. of (\ref{eq:prRN-3}). As explained above, this will be used to show lower bounds on the spectrum of the Hamilton operator. For $\g \in (0 ; 1/4)$ we introduce the notation \[ F_p^\g= (1 - C N^{-\g}) p^2 +\widehat{\o}_N(p) \] and recall the definition of $G_p$ in \eqref{eq:defFpGp}. For $p \in \L^*_+$ we consider the coefficient $\upsilon_p$ defined through  
\be \label{eq:alphap}
\tanh(2 \upsilon_p) = \a_p := \frac 1 {G_p} \Big( F_p^\g - \sqrt{(F_p^\g)^2 - G_p^2}\Big)\,,
\ee
and the antisymmetric operator 
\be \label{eq:B-sigma}
B_\upsilon= \frac 12 \sum_{p \in P_L} \upsilon_p (b^*_p b^*_{-p} - b_p b_{-p})\,.
\ee
With Lemma \ref{lm:boundsFG}, it is easy to check that $|\a_p|<1$ hence $\upsilon_p$ is well defined. Moreover, we have $\|\upsilon\|_2 \leq C$ 
, hence in particular 
\be \begin{split}\label{eq:grow-Bups}
e^{-B_\upsilon}(\cN_++1)^ke^{B_\upsilon} & \leq C (\cN_++1)^k \\
\end{split}\ee 
for a constant $C>0$ (depending on $k$),proceeding as in the proof of Lemma \ref{lm:growcNcH}. 

\smallskip

{\it Remark.}  The choice (\ref{eq:alphap}) is motivated by the lower bound \eqref{eq:complete} because, up to negligible errors, $e^{B_v} b_p^* b_p e^{-B_v} = c_p^* c_p$, with $c_p, c_p^*$ defined in \eqref{eq:defcp} and satisfying canonical commutation relations.

\begin{prop} \label{prop:BsigmaRN-LB}Let $V \in L^3(\bR^2)$ as in Theorem \ref{thm:main}, $B_\upsilon$ as in \eqref{eq:B-sigma}, and let $\nu \in (1/6;1/2)$. Then, for $N \in \bN$ large enough, and any $\g \in (0,1/4)$, there exists a constant $C>0$ s.t.  
such that 
	\begin{equation}\label{eq:lbRN}\begin{split}
			e^{-B_\upsilon} \cR_N e^{B_\upsilon} \geq &\; E_N^{\text{Bog}}   +(1- C N^{-\g})  \cD_\g   \\
		  & - C (\log N) \big[ N^{\g-1} (\cN_++1)^2  + N^{-\g}(\cN_++1)\big]\,,
	\end{split}\end{equation}
with $E_N^\text{Bog}$ as defined in (\ref{eq:EcM}) and where $D_\g$ is the quadratic operator
	\[
\cD_\g =\sum_{\substack{p \in P_\g}} \sqrt{p^4 + 8\pi p^2}\,  a^*_p a_p + \tfrac {1 } 2  N^\g \hskip -0.2cm \sum_{ p\in  \L_+^*\setminus P_\g}a^*_p a_p \, \]
where $P_\g = \{ p \in \L^*_+ : |p| \leq N^{\g/2}\}$. 
\end{prop}

\begin{proof} For $p \in P_L$, we  introduce the notation 
\be \label{eq:defcp}
 c_p = \frac{b_p + \a_p b^*_{-p}}{\sqrt{1-\a_p^2}}\,,
\ee
with $\a_p$ defined before  \eqref{eq:B-sigma}. 
A standard completion of the square argument (see  \cite[Sect.\,3]{Sei}, so as  \cite[Thm. 6.3]{LS}) leads to the lower bound
\begin{equation}\label{eq:complete} \begin{split}
&F_p^\g \big(b^*_p b_p + b^*_{-p} b_{-p}\big) + G_p \big(b_p b_{-p} + b^*_p b^*_{-p} \big) \\[6pt]
& \; \geq \sqrt{(F_p^\g)^2 - G_p^2} \,\big(c^*_p c_p + c^*_{-p} c_{-p}\big)  - \frac 12 \Big( F_p^\g - \sqrt{(F_p^\g)^2 - G_p^2}   \Big) \Big([b_p, b^*_p] + [b_{-p}, b^*_{-p}]\Big)\,.
\end{split}\end{equation} 
Expanding the square roots we find
\[ \begin{split}
\Big|\; F_p^\g - \sqrt{(F_p^\g)^2 - G_p^2}\;\Big| & \leq C \frac{|\widehat \o_N(p)|^2}{|p|^2} \\
\Big|\; F_p^\g - \sqrt{(F_p^\g)^2 - G_p^2} - F_p + \sqrt{F_p^2 - G_p^2} \;\Big| &\leq C N^{-\g} \frac{|\widehat \o_N(p)|^2}{|p|^2} 
\end{split}\]
for all $p \in P_L$. Hence, we find 
\[ \begin{split}
\Big|\sum_{p \in P_L}  \Big( F_p^\g - \sqrt{(F_p^\g)^2 - G_p^2} \Big) \, \Big|& \leq C (\log N) \,,
\end{split}\]
and, similarly,
\[ \begin{split}
\Big|\sum_{p \in P_L}  \Big( F_p^\g - \sqrt{(F_p^\g)^2 - G_p^2} - F_p + \sqrt{F_p^2 - G_p^2} \Big)\Big|
\leq C N^{-\g} (\log N) \,.
\end{split}\]
With the commutation relations \eqref{eq:bcomm}, we conclude therefore that 
\[ \begin{split}
 - \frac 12 \sum_{p\in P_L}& \Big( F_p^\g - \sqrt{(F_p^\g)^2 - G_p^2}   \Big) \Big([b_p, b^*_p] + [b_{-p}, b^*_{-p}]\Big) \\
 &\geq  - \frac 12  \sum_{p\in P_L}\Big( F_p - \sqrt{F_p^2 - G_p^2}   \Big)   - C N^{-\g} (\log N) (\cN_++1)\,.
\end{split}\]
Next we consider the first term on the r.h.s. of (\ref{eq:complete}). We  denote $\cE^\g(p)=\sqrt{(F_p^\g)^2 - G_p^2} $. Proceeding as in \eqref{eq:BogSpectrum}, we find
\[
\big|\,\cE^\g(p)- \sqrt{|p|^4 + 8\pi |p|^2}\,\big| \leq  C (N^{-\g} p^2 + N^{-\a}|p| + N^{-1} ( \log N))\,.
\]
Hence  
\be \begin{split} \label{eq:cE-gamma}
\sum_{p \in P_L} \cE^\g(p) c^*_p c_p & \geq   \sum_{p\in P_L} \sqrt{|p|^4 + 8\pi |p|^2}\, c^*_p c_p -  C N^{-\g} \sum_{p\in  P_L}  |p|^2 c^*_p c_p  \\
&\geq   ( 1- C N^{-\g})  \sum_{p \in P_L} \cE(p)\, c^*_p c_p\,,
\end{split}\ee
where we introduced the notation $\cE(p) = \sqrt{|p|^4 + 8\pi |p|^2}$, for $p \in P_\gamma= \{ p \in \L^*_+ : |p| < N^{\gamma/2} \}$, and $\cE (p) = N^\g$, for $p \in P_L \backslash P_\gamma$. Next, we use (see \cite[Lemma 5.3]{BBCS4} for a proof) that  
\be \label{eq:actionBv}
e^{B_\upsilon} b_p  e^{-B_\upsilon} = \cosh(\upsilon_p) b_p + \sinh(\upsilon_p) b^*_{-p} + D_p\,,
\ee
where the remainder operator $D_p$ satisfies
\be \label{eq:Dp-bound}
\| (\cN_++1)^{n/2} D_p \xi\| \leq \frac C N |\upsilon_p| \| (\cN_++1)^{(n+3)/2} \xi\| + \frac C N \int_0^1 ds \| a_p (\cN_++1)^{(n+2)/2} e^{s B_\upsilon} \xi\|\,.
\ee
This implies that, after some algebraic manipulations, that 
\[
 \sum_{p \in P_L} \cE(p) c^*_p c_p =   \sum_{p \in P_L} \cE(p) e^{B_\upsilon} b^*_p b_p  e^{-B_\upsilon} +  \d_N^{(1)} + \d_N^{(2)}
\]
with
\[
\d_N^{(1)} =  \sum_{p \in P_L} \cE(p) D^*_p e^{B_\upsilon} b_p  e^{-B_\upsilon}\,, \qquad \d_N^{(2)}  =  \sum_{p \in P_L} \cE(p) \big( \cosh(\s_p) b_p + \sinh(\s_p) b^*_{-p}\big) D^*_p\,.
\]
With \eqref{eq:Dp-bound} and \eqref{eq:grow-Bups} we easily bound
\[ \begin{split}
\big | \langle \xi, \d_N^{(1)} \xi \rangle \big|
 & \leq \Big(\sup_{p \in P_L} \cE(p) \Big) \sum_{p \in P_L} \| (\cN_++1)^{-1/2} D_p \xi\|  \| (\cN_++1)^{1/2}   e^{B_\upsilon} b_p  e^{-B_\upsilon}\xi\| \\
 & \leq C N^{\g-1} \sum_{p \in P_L} |\upsilon_p| \|(\cN_++1) \xi\| \| a_p  e^{-B_\upsilon}\xi\|\leq C N^{\g-1} \|(\cN_++1) \xi\|^2\,.
\end{split}\]
where we used  \eqref{eq:grow-Bups}. The term $\d_N^{(2)}$ can be bounded similarly. Using again  \eqref{eq:grow-Bups}, we also obtain 
\be\label{eq:diff-bb-aa}
\pm\sum_{p \in P_L}  \Big[e^{B_\upsilon} b^*_p b_p  e^{-B_\upsilon}-  e^{B_\upsilon} a^*_p a_p  e^{-B_\upsilon}\Big]\leq CN^{-1}(\cN_++1)^2\,.\ee
Summarizing, from (\ref{eq:complete})-(\ref{eq:diff-bb-aa}) we obtain that  
\be\label{eq:cR-LB}\begin{split}
&\sum_{p \in P_L} \big((1 - C N^{-\g}) p^2 + \widehat \o_N(p)\big) b^*_p b_p + \frac 12 \sum_{p\in P_L} \widehat{\omega}_N (p) \big[ b^*_p b^*_{-p} + b_p b_{-p} \big] \\
& \; \geq  - \frac 12  \sum_{p \in P_L} \Big( F_p - \sqrt{F_p^2 - G_p^2}   \Big)  + (1 - C N^{-\g}) \sum_{p \in P_\g} \sqrt{p^4 + 8\pi p^2}\, e^{B_\upsilon} a^*_p a_p e^{-B_\upsilon} \\
& \hskip 0.5cm \;+ \tfrac 1 2 N^{\g} \sum_{p \in P_L \setminus P_\g}e^{B_\upsilon} a^*_p a_p e^{-B_\upsilon}- C N^{\g-1} (\cN_++1)^2  -  C(\log N) N^{-\g} (\cN_++1)\,.
\end{split}\ee
Inserting on the r.h.s. of \eqref{eq:prRN-3}, we find 
\[\begin{split}
		\cR_N \geq &\; \cC_\cR - \frac 12  \sum_{p \in P_L} \Big( \sqrt{F_p^2 - G_p^2}  -F_p \Big)  +  \tfrac 1 2 N^{\g} \hskip -0.2cm \sum_{p \in \L^*_+\setminus P_L} a^*_pa_p \\ 
		&+ (1 - C N^{-\g}) \sum_{p \in P_\g} \sqrt{p^4 + 8\pi p^2}\, e^{B_\upsilon} a^*_p a_p e^{-B_\upsilon} +   \tfrac 1 2 N^{\g} \hskip -0.2cm\sum_{p \in P_L \setminus P_\g} e^{B_\upsilon} a^*_p a_p e^{-B_\upsilon}   \\
		&  - C N^{\g-1} (\cN_++1)^2    - C (\log N) N^{-\g} (\cN_++1)\,.
\end{split}\]
With Eq. \eqref{eq:const-last} and \eqref{eq:Iell} (choosing $\ell=\frak{a}$) 
and using the a-priori bound \eqref{eq:grow-Bups} we have
\begin{equation}\label{eq:lbRN-fin}\begin{split}
	e^{-B_\upsilon} \cR_N e^{B_\upsilon} \geq &\; E_N^{\text{Bog}}   + (1-CN^{-\g}) \sum_{\substack{p \in P_\g}} \sqrt{p^4 + 8\pi p^2}\,  a^*_p a_p  \\
	& +   \tfrac 1 2  N^{\g}\sum_{p \in P_L \setminus P_\g}a^*_pa_p  + \tfrac 1 2 N^{\g} \hskip -0.2cm \sum_{p \in \L^*_+ \setminus P_L} e^{-B_\upsilon} a^*_p a_p e^{B_\upsilon} \\
	&   - C (\log N) \big[ N^{\g-1} (\cN_++1)^2  + N^{-\g}(\cN_++1) \big]\,.
\end{split}\end{equation}
Observing that, by (\ref{eq:bcomm}) (in particular, the last two commutators), $[B_v , a_p^* a_p ] = 0$ for all $p \in \L^*_+ \backslash P_L$ (because, from (\ref{eq:B-sigma}), $B_v$ only contains the operators $b_p, b_p^*$ with $p \in P_L$), we arrive at (\ref{eq:lbRN}). 
\end{proof}

\section{Proof of Theorem \ref{thm:main}.} \label{sec:proof} 

In this section we focus on the low energy spectrum of $H_N$. We fix $\alpha =5/2$ and $\nu = 1/5$ (recall the definitions of $\ell=N^{-\a}$ and $P_L= \{ p \in \L^*_+: |p|\leq N^{\a + \nu}\}$ entering in the definitions of the operators $B$, $A$ and $D$ defined in \eqref{eq:defB}, \eqref{eq:defA} and \eqref{eq:defD} respectively).  

First of all, we observe that, from Prop. \ref{prop:expD-cRN} (choosing $\xi_L = e^{B_\tau} \Omega$, with $B_\tau$ defined as in (\ref{eq:defBt})) and Prop. \ref{prop:gsandspectrum}, the ground state energy $E_N$ satisfies  
\begin{equation}\label{eq:EN-UB} E_N \leq E_N^\text{Bog} + C N^{-3/10 + \delta} \end{equation} 
for any $\delta >0$, if $N$ is large enough. Recall here the definition (\ref{eq:EcM}) of $E_N^{\text{Bog}}$. 

Next, we prove lower bounds for the ground state energy and for the excited eigenvalues of $H_N$ below the threshold $E_N + \zeta$. For $k \in \bN$, let $\l_k$ be the $k$-th eigenvalue of $H_N-E_N^{\text{Bog}}$ and $\m_k$ the $k$-th eigenvalue of the quadratic operator 
\be \label{eq:cD}
\cD_\gamma = \sum_{p \in P_\g} \sqrt{|p|^4 + 8\pi p^2} a^*_p a_p + \frac{N^\gamma}{2} \sum_{p \in \L^*_+ \backslash P_\g} a_p^* a_p 
\ee  
with $P_\g = \{ p \in \L^*_+ : |p| \leq N^{\g/2} \}$, as defined in Prop. \ref{prop:BsigmaRN-LB} (note that eigenvalues are counted with multiplicity). We claim that 
\be \label{eq:diff-lk-nuk}
\l_k \geq  \m_k - C  N^{-1/10+ \delta} (1 + \zeta)^2
\ee
for all $k \in \bN$ with $\l_k < \zeta + 1$ and for any $\delta > 0$. In view of the upper bound (\ref{eq:EN-UB}), this bound is enough to show (\ref{eq:EN}) (taking $k=0$) and to prove lower bounds matching (\ref{eq:exc-main}) for all eigenvalues of $H_N - E_N$ below the threshold $\zeta$, if $N$ is large enough. Here we use the fact that the spectrum of the quadratic operator (\ref{eq:cD}) below a fixed $\zeta > 0$ consists exactly of eigenvalues having the form 
\[ \sum_{p \in \L^*_+} n_p \sqrt{|p|^4 + 8\pi p^2} \]
with $n_p \in \bN$ for all $p\in \L^*_+$ and $n_p \not = 0$ for finitely many $p \in \L^*_+$ only (to stay below the threshold $\zeta >0$, we cannot excite modes with $|p| > N^{\gamma/2}$). 

To prove (\ref{eq:diff-lk-nuk}), we apply a localization argument similar to those recently used in \cite{NT,HST}, together with the a-priori bound on the energy of excitations established in Prop. \ref{lm:hpN}. Let $\cL_N$ and $B$ be defined in \eqref{eq:calL} and  \eqref{eq:defB} respectively, and consider the excitation Hamiltonian
 \be \label{eq:calG-1}
 \cG_N=e^{-B}\cL_N e^B\,.\
 \ee
From \cite[Eq. (61)]{CCS}, we have the condensation bound
\begin{equation}\label{eq:cGN-fin} 
\cG_{N} - E_N^\text{Bog} \geq c\, \cN_+ - C  
\end{equation}
for all $N \in \bN$ sufficiently large (here, we used that $2\pi N \geq  E_N^\text{Bog}$). We will make use of the following lemma, which is proven in App.~\ref{App:propRN}.  
\begin{lemma} \label{GN-localization}
		Let $V\in L^3(\bR^2)$ as in Theorem \ref{thm:main}.    Let $\cG_N$ be defined as in \eqref{eq:calG-1}. Let  $f, g:\bR \to [0;1]$ be smooth functions, with $f^2 (x) + g^2 (x)= 1$ for all $x \in \bR$. Moreover, assume that $f (x) = 0$ for $x > 1$ and $f (x) = 1$ for $x < 1/2$. For a small $\eps>0$, we fix $M  = N^{\e}$  and we set $f_M (\cN_+)= f (\cN_+ / M), g_M(\cN_+) = g (\cN_+ / M)$. Then there exists a constant $C>0$ such that 
\begin{equation}\label{eq:cGN-1}\begin{split} \cG_{N}  - E_N^\text{Bog}  \geq\; & f_M (\cN_+)\big(\cG_N -E_N^\text{Bog} \big) f_M(\cN_+) + g_M(\cN_+) \big(\cG_N-E_N^\text{Bog} \big) g_M (\cN_+) - \cE_M \end{split}
\end{equation}
with 
\be \label{eq:cEM-loc}
\cE_M \leq C \frac{N^{1/2}}{M^2}[\|f'\|^2_\infty + \|g'\|^2_\infty](\cH_N+1)\,,\ee
for all $\a>1$, $M \in \bN$ and $N \in \bN$ large enough. 
\end{lemma}

Let now $Y \subset \cF_+^{\leq N}$ denote the subspace spanned by the eigenvectors of $\cG_N - E_N^\text{Bog}$ associated with its first $(k+1)$ eigenvalues $\lambda_0 \leq \lambda_1 \leq \dots \leq \lambda_k$. Since, by assumption, $\lambda_k \leq \zeta +1$, we find $Y \subset P_{\zeta+1} (\cF_+^{\leq N})$. We have 
\begin{equation}\label{eq:lblambda} \begin{split} \lambda_k &= \sup_{\substack{\xi \in Y\\\|\xi\|=1}} \langle \xi , (\cG_N - E_N^\text{Bog}) \xi \rangle \\ &\geq  \sup_{\substack{\xi \in Y\\\|\xi\|=1}} \langle \xi, [f_M(\cN_+)(\cG_N-E_N^{\text{Bog}})f_M(\cN_+) + g_M(\cN_+)(\cG_N-E_N^{\text{Bog}})g_M(\cN_+) -\cE_M]\xi\rangle\,.\end{split} 
 \end{equation}
 where $M = N^\epsilon$ for a $\epsilon > 0$ to be specified below (we will choose $\epsilon = 3/4 + 1/20$). From Eq. \eqref{eq:cGN-fin}  we have that, for $N$ large enough (recall $M = N^\epsilon$), 
 \begin{equation}
 	\label{eq:loc1}
 	\begin{split}
 	g_M(\cN_+)(\cG_N-E_N^{\rm{Bog}})g_M(\cN_+) & \geq g_M^2(\cN_+) (c\,\cN_+ - C ) \geq  g_M^2(\cN_+)\big(c' M-C\big) \geq 0\,,
 	\end{split} 
 \end{equation}
since $g_M=0$ for $\cN_+ \leq M/2$. Furthermore, for a normalized $\xi \in Y \subset P_{\zeta+1} (\cF_+^{\leq N})$, we find, combining \eqref{eq:cEM-loc} with Prop. \ref{lm:hpN} and Lemma \ref{lm:ANgrow}, 
\begin{equation}\label{eq:loc3}
\begin{split}
\langle \xi,\cE_M\xi\rangle &\leq C \frac{ N^{1/2}}{M^2} \langle e^{-A}\xi, e^{-A}(\cH_N+1)e^Ae^{-A} \xi\rangle \leq C \frac{ N^{3/2}}{M^2}(\log N)(1+\z)\,, \end{split}
\end{equation}
where we used again $2\pi N \geq E_N^{\text{Bog}}$ to make sure that $e^{-A}\xi$ satisfies the assumptions of Prop. \ref{lm:hpN}. 

Finally, we look at the first term on the r.h.s. of \eqref{eq:lblambda}. With \eqref{eq:lbRN} we find
\begin{equation*}
	\begin{split}
		  & \langle \xi, e^{A}e^{B_\upsilon} f_M(\widetilde \cN_+)e^{-B_\upsilon}e^{-A}(\cG_N-E_N^{\rm{Bog}}) e^{A}e^{B_\upsilon}  f_M(\widetilde\cN_+)e^{-B_\upsilon}e^{-A}  \xi\rangle\\
		  &\hspace{3cm}\geq \langle \xi, e^{A}e^{B_\upsilon} f_M(\widetilde \cN_+)((1- C N^{-\g})\cD_\g -\cE_\g) f_M(\widetilde \cN_+)e^{-B_\upsilon}e^{-A}  \xi\rangle\\
	\end{split}
\end{equation*}
 where we introduced the notation $\widetilde \cN_+ := e^{-B_\upsilon}e^{-A} \cN_+ e^{A}e^{B_\upsilon}$ and 
\[\begin{split} \cE_\g \leq  &\,C (\log N) \big[ N^{\g-1} (\cN_++1)^2 + N^{-\g} (\cN_++1)\big]\,.
	\end{split}\]
Now, with Lemma \ref{lm:ANgrow} and Eq.\eqref{eq:grow-Bups} 
we have
\[\begin{split}
 f_M(\widetilde\cN_+)(\cN_++1)^2 f_M(\widetilde\cN_+) &\leq CM f_M(\widetilde\cN_+)(\widetilde\cN_++1)f_M(\widetilde\cN_+) \leq CM (\widetilde\cN_++1)\,. 
 \end{split}\]
Hence, for any normalized $\xi \in Y \subset P_{\zeta+ 1} (\cF^{\leq N}_+)$, we find, with (\ref{eq:cGN-fin}), 
\[ \begin{split}  \langle \xi, e^{A}e^{B_\upsilon} f_M(\widetilde \cN_+) \cE_\g f_M (\widetilde \cN_+) &e^{-B_\upsilon}e^{-A}  \xi\rangle \\ &\leq C (\log N) (MN^{\gamma-1} + N^{-\gamma}) \langle \xi, (\cN_+ + 1) \xi \rangle \\ &\leq C (\log N) (MN^{\gamma-1} + N^{-\gamma}) \langle \xi, (\cG_N - E_N^\text{Bog} + C) \xi \rangle \\ &\leq C (\log N) (\zeta + 1)(MN^{\gamma-1} + N^{-\gamma}) \,. \end{split} \]
We conclude that 
\[ \begin{split} \lambda_k \geq \; & (1-C N^{-\g}) \sup_{\substack{\xi \in Y\\\|\xi\|=1}} \langle \xi, e^A e^{B_\upsilon} f_M(\widetilde \cN_+) \cD_\g f_M (\widetilde \cN_+) e^{-B_\upsilon}e^{-A}  \xi\rangle \\ &- C (\log N) (\zeta + 1)(MN^{\gamma-1} + N^{-\gamma}) - C N^{3/2} M^{-2} (\log N) (\zeta + 1)\,. \end{split} \]
Next we observe that, for any normalized $\xi \in Y \subset P_{\zeta + 1} (\cF_+^{\leq N})$, we have (again, with (\ref{eq:cGN-fin})) 
\begin{equation}\label{eq:fM-markov} \| f_M (\widetilde \cN_+) e^{-B_\upsilon}e^{-A}  \xi \|^2 \geq 1 - \frac{C}{M} \langle \xi , \cN_+  \xi  \rangle  \geq 1 - \frac{C (\zeta+1)}{M}\,. \end{equation} 
This immediately implies that the linear subspace $X = f_M (\widetilde{\cN}_+) e^{-B_\upsilon}e^{-A} Y \subset \cF_+^{\leq N}$ has dimension $(k+1)$ (like $Y$) and that 
\[ \begin{split} \lambda_k \geq \; &(1-C N^{-\g} - C (\zeta+1) M^{-1}) \sup_{\xi \in X , \|\xi\|=1} \langle \xi, \cD_\g  \xi\rangle \\ &- C (\log N) (\zeta + 1)(MN^{\gamma-1} + N^{-\gamma}) - C N^{3/2} M^{-2} (\log N) (\zeta + 1) \,. \end{split} \]
Thus, by the min-max principle for the eigenvalues of $\cD_\g$, 
\begin{equation} \label{eq:l-mu} \begin{split} \l_k \geq \; &(1-C N^{-\g} - C (\zeta+1) M^{-1}) \mu_k \\ &- C (\log N) (\zeta + 1)(MN^{\gamma-1} + N^{-\gamma}) - C N^{3/2} M^{-2} (\log N) (\zeta + 1)\,. \end{split}  \end{equation} 
Choosing $M = N^{3/4+1/20}$ and $\g = 1/4 - 3/20$, we obtain (using that (\ref{eq:l-mu}) in particular implies that $\mu_k \leq C \lambda_k \leq C (\zeta+1)$)   
\be \label{eq:low-lk-3}
\l_\k \geq \mu_k  - C N^{-1/10+\d} (1+\zeta)^2
\ee
for any $\d>0$. 

Finally, we show upper bounds for all the excited eigenvalues $\lambda_k$ of $H_N - E_N^{\text{Bog}}$ (or equivalently of $\cR_N - E_N^{\text{Bog}}$) with $\lambda_k \leq \zeta + 1$ (we already proved an upper bound for the ground state energy, with $k=0$, at the beginning of this section). We are going to use trial states given by eigenvectors of the operator $\cD_\g$, defined in \eqref{eq:cD}. Fix $k \in \bN \backslash \{ 0 \}$, with $\lambda_k < \zeta$. For $j=1,\dots ,k$, the $j$-th eigenvalue $\mu_j$ of $\cD$ has the form 
\[ \mu_j = \sum_{p \in P_L} n_p^{(j)} \eps_p \]
with $\eps_p = \sqrt{|p|^4 + 8\pi p^2}$ and $n_p^{(j)} \in \bN$, for all $p \in P_L$ (since we consider eigenvalues below a fixed $\zeta> 0$, there is no contribution from the second sum in (\ref{eq:cD}), running over $p \in \L_+^* \backslash P_\g$, and there are only finitely many $p \in P_L$ with 
$n_p^{(j)} \not = 0$). The eigenvector associated with $\mu_j$ has the form
\begin{equation}
	\label{eq:eigenvectorD}
	\xi_j = C_j \prod_{p\in P_L}(a_p^*)^{n_p^{(j)}}  \Omega,
\end{equation}
for an appropriate normalization constant $C_j > 0$ (if the eigenvalue has multiplicity larger than one, eigenvectors are not uniquely defined, but they can always be chosen in this form). We denote by $\text{span} (\xi_1, \dots , \xi_k)$ the linear space spanned by the eigenvectors defined in (\ref{eq:eigenvectorD}). From the min-max principle, we have 
\[ \begin{split}
\l_k &=  \inf_{\substack{Y \subset \cF_+^{\leq N}\\ \dim Y =k}} \sup_{\substack{\xi \in Y\\\|\xi\|=1}} \langle  \xi, e^{-B_\t} e^{-D}(\cR_N-E_N^{\text{Bog}})e^D e^{B_\t} \xi\rangle \\
&\leq \sup_{\substack{\xi \in \text{ span}(\xi_1, \ldots, \xi_k)\\ \|\xi\|=1}}\langle \xi, e^{-B_\t}e^{-D}(\cR_N-E_N^{\text{Bog}})e^D e^{B_\t} \xi\rangle \,.
\end{split}\]
Since $a_p e^{B_\t} \xi = 0$ for all $p \in P_L^c$ and all $\xi \in \text{ span} (\xi_1, \dots , \xi_k)$, we can apply Prop. \ref{prop:expD-cRN} with $\kappa = 10$ (so that $\kappa > 4 (\alpha+\nu -1/2)$) to conclude that  (recall the definition \eqref{eq:RNBog})
\[ \begin{split}
	\l_k 	&\leq \; \sup_{\substack{\xi \in\, \text{Span}(\xi_1,\dots, \xi_k)\\\|\xi\|=1}} \Big[ \langle \xi, \big(e^{-B_\t}\cR^{\text{Bog}}_N e^{B_\t}-E_N^{\text{Bog}} \big)\xi\rangle \\
& \hskip 3.5cm +CN^{-3/10 +\d} \langle \xi, e^{-B_\t} \cK (\cN_++1)^{15} e^{B_\t} \xi \rangle \Big]\,,
\end{split}\]
for any $\delta > 0$. With  Prop. \ref{prop:gsandspectrum} and Lemma \ref{lm:growcNcH} we get
\[ \begin{split}
	\l_k  &\leq \; \sup_{\substack{\xi \in\, \text{Span}(\xi_1,\dots, \xi_k)\\\|\xi\|=1}} \Big[ \langle \xi, \cD \xi\rangle +C N^{-3/10 +\d} \langle \xi, (\cH_N +1) (\cN_++1)^{15}  \xi \rangle \Big] 
\end{split}\]
for any $\delta > 0$ (the value of $\delta$ changes from line to line). Observing that, on $\text{span} (\xi_1, \dots , \xi_k)$, 
\be \begin{split} \label{eq:cV-on-xiL}
		\bmedia{\xi, \cV_N  \xi}  & \leq  C \sum_{\substack{r \in \L^*,\, p,q \in P_L \\ r+p , q-r \in P_L}} |\widehat V(r/e^N)| \| a_{p+r} a_{q} \xi \| \| a_p a_{q+r} \xi \| \\
		& \leq C   \sum_{\substack{r \in \L^*,\, p,q \in P_L \\ r+p , q-r \in P_L}} \frac{1}{|p|^2} \,|q|^2 \| a_{p+r} a_{q} \xi \|^2 \leq C (\log N) \| \cK^{1/2}\cN_+^{1/2} \xi\|^2
	\end{split}\ee
and, again on $\text{span} (\xi_1, \dots , \xi_k)$, $\cN_+ \leq C \cK \leq C \cD_\g \leq C \mu_k \leq C (\zeta + 1)$  (from the lower bound (\ref{eq:diff-lk-nuk}), we have $\mu_k \leq \l_k + 1 \leq \zeta + 1$, for $N$ large enough), we conclude that 
\[ \langle \xi, (\cH_N + 1)( \cN_+ + 1)^{15} \xi \rangle \leq C (\log N) \langle \xi, (\cK+1)(\cN_+ + 1)^{16} \xi \rangle \leq C (\log N) (1 + \zeta)^{17} \]
for all normalized $\xi \in \text{span} (\xi_1, \dots , \xi_k)$. Thus, we find 
\[\begin{split}
	\l_k &\leq \sup_{\substack{\xi \in\, \text{Span}(\xi_1,\dots, \xi_k)\\\|\xi\|=1}} \langle \xi, \cD \xi\rangle + C N^{-3/10+\delta} (1+\z)^{17}  \leq \m_k + C N^{-3/10+\delta} (1+\z)^{17}\,. 
\end{split}\]
Together with the lower bound (\ref{eq:diff-lk-nuk}), this concludes the proof of (\ref{eq:exc-main}) and of Theorem~\ref{thm:main}.


\appendix



\section{Proof of Proposition \ref{prop:RN} and of Lemma \ref{GN-localization}} 
\label{App:propRN}

To show Prop. \ref{prop:RN}, we start from the analysis carried out in \cite[Sec.6 and App.A]{CCS}. First, we study the excitation Hamiltonian 
\be \label{eq:calG}
\cG_{N}:=  e^{-B} \cL_N e^{B} = \sum_{i=1}^4\cG_{N}^{(i)}\,,\qquad  \cG_N^{(i)}:= e^{-B}\cL_N^{(i)} e^B
\ee
with $\cL_N^{(i)}$  defined in \eqref{eq:cLNj} and $B$ defined \eqref{eq:defB}. 

\begin{prop} \label{prop:GN} Let $V\in L^3(\bR^2)$ be compactly supported, pointwise non-negative and spherically symmetric. Let $\cG_{N}$ and $\widehat{\o}_N(p)$ be defined in \eqref{eq:calG} and \eqref{eq:defomegaN} respectively. Then for $\ell=N^{-\a}$ there exists a constant $C>0$ such that 
	\be \label{eq:GNeff} \begin{split}   \cG_{N}  :=\; &    \frac N2 \big(\widehat{V}(\cdot/e^N)*\widehat{f}_{N,\ell}\big)(0) (N-1)\left(1-\frac{\cN_+}{N}\right)  + \frac 12\sum_{p\in \L^*_+} \widehat{\o}_N(p)\eta_p\\
			& + \left[2 N\widehat{V} (0)-\frac N2 \big(\widehat{V}(\cdot/e^N)*\widehat{f}_{N,\ell}\big)(0) \right] \, \cN_+ \, \left(1-\frac{\cN_+}{N}\right) \\
			&+ \frac 1 2 \sum_{ p\in  \L_+^*}\widehat{\omega}_N(p)(b_pb_{-p}+\hc)   + \sqrt{N} \sum_{\substack{p,q \in \L^*_+ :\\ p + q \not = 0}} \widehat{V} (p/e^N) \left[ b_{p+q}^* a_{-p}^* a_q  + \hc \right]  \\
			& +\cH_N + \cE_\cG \,,
	\end{split} \ee
	with $\cE_\cG$ satisfying
	\be \label{eq:GeffE}
	\begin{split}
		| \langle \xi, \cE_{\cG}\, \xi \rangle | \leq\; & C \big( N^{1/2 -\a} + N^{-1}(\log N)^{1/2} \big) \| \cH_N^{1/2}\xi\| \| (\cN_++1)^{1/2}\xi \| \\
		&+ C N^{1-\a}\|(\cN_++1)^{1/2}\xi \| ^2\,,
	\end{split}\ee
	for $\a>1$, and $N \in \bN$ sufficiently large.
	
\end{prop}

\begin{proof}[Proof of Prop. \ref{prop:GN}] The proof of  Prop. \ref{prop:GN} follows from the analysis performed in \cite[App. A]{CCS}, where we establish properties of the operators $\cG_N^{(i)}= e^{-B}\cL_N^{(i)} e^B$.  Recombining the results of \cite[Prop. 13--16]{CCS} we have  
	\begin{equation} \begin{split} \label{eq:proofGNell-1}
			& \cG_{N} =    \frac{\widehat{V} (0)}{2}\, (N +\cN_+ -1) \, (N-\cN_+)\\
			& \hskip 0.2cm+ \sum_{p \in \L^*_+} \eta_p \Big[p^2 \eta_p + N\widehat{V} (p/e^N) + \frac 12\sum_{\substack{r \in \L^*\\ p+r \neq 0}} \widehat{V} (r/e^N)  \eta_{p+r}\Big]\Big(\frac{N-\cN_+}{N}\Big) \Big(\frac{N-\cN_+ -1}{N}\Big)\\
			&  \hskip 0.2cm+\cK +N\sum_{p \in \Lambda^*_+}  \widehat{V} (p/e^N) a^*_pa_p\Big(1-\fra{\cN_+}{N}\Big) \\
			&  \hskip 0.2cm+ \sum_{p \in \L^*_+} \Big[\; p^2 \eta_p + \frac N2 \widehat{V} (p/e^N) + \frac 12 \sum_{r \in \L^*:\; p+r \neq 0}\hskip -0.5cm \widehat{V} (r/e^N)  \eta_{p+r} \; \Big]  \big( b^*_p b^*_{-p} + b_p b_{-p} \big) \\
			&  \hskip 0.2cm+ \sqrt{N} \sum_{p,q \in \L^*_+ :\, p + q \not = 0} \widehat{V} (p/e^N) \left[ b_{p+q}^* a_{-p}^* a_q  + \hc \right]  +\cV_N   + \cE_{1}
	\end{split} \end{equation}
	with $\cK$ and $\cV_N$ defined in \eqref{eq:cKVN}, and where 
	\begin{equation*}
		| \langle \xi, \cE_1  \xi \rangle | \leq  C N^{1/2 -\a} \| \cH_N^{1/2}\xi\| \| (\cN_++1)^{1/2}\xi \| + C N^{1-\a}  \| (\cN_++1)^{1/2}\xi \|^2 
	\end{equation*}
	for any $\a >1$ and  $\xi \in \cF^{\leq N}_+$. With the  scattering equation \eqref{eq:eta-scat}, we rewrite
	\begin{equation*} 
		\begin{split} 
			&\sum_{p \in \L^*_+}  \eta_p \Big[p^2 \eta_p + N\widehat{V} (p/e^N) + \frac 1 {2} \sum_{r \in \L^*:\; p+r \neq 0}\hskip -0.5cm \widehat{V} (r/e^N)  \eta_{p+r}\Big] \\
			& = \sum_{p \in \L^*_+} \eta_p \Big[ \;\frac N 2 \widehat{V} (p/e^N)  +Ne^{2N} \l_\ell \widehat \chi_\ell(p) +e^{2N} \l_\ell \sum_{q \in \L^*} \widehat \chi_\ell(p-q) \eta_q - \frac 12 \widehat V(p/e^N) \eta_0\;\Big] 
	\end{split} \end{equation*}
With the bounds $|e^{2N}\l_\ell|\leq C N^{2\a-1}$, $\|\eta\| \leq C N^{-\a}$ and $\| \widehat{\chi}_\ell * \eta \| = \| \chi_\ell \check{\eta} \| \leq \| \check{\eta} \| 	\leq C N^{-\a}$, 	
we estimate 
	\[
\Big| e^{2N} \l_\ell \sum_{\substack{p \in \L^*_+,\, q \in \L^*}} \widehat \chi_\ell(p-q) \eta_q\eta_p  \Big|  \leq C N^{2\a-1} \big( \| \chi_\ell\| + \| \hat \chi_\ell \ast \eta \|\big)\| \eta\| \leq CN^{-1}.
\]
On the other hand, using that $\sum_{ p\in  \L_+^*}|\widehat{V}(p/e^N)|/|p|^2\leq CN$ and the bound $|\eta_0| \leq C N^{_2\alpha}$ (see Lemma \ref{lm:eta}), we have
	\[
	\Big| \frac 12 \sum_{p \in \L^*_+} \widehat V(p/e^N) \eta_p \eta_0 \Big| \leq C N^{1-2\a}\,.
	\]
	Finally,  by the definition \eqref{eq:defeta} of $\eta_p$ and using again $|\eta_0| 
	\leq C N^{-2\alpha}$ (since we need to add the zero momentum mode), we rewrite
	\[ \begin{split}
		\frac N 2 \sum_{p \in \L^*_+}  \widehat{V} (p/e^N) \eta_p 
		=\; & \frac{N^2}2 \Big( \widehat{V}(\cdot/e^N)*\widehat{f}_{N,\ell}\big)(0)-  \widehat V(0) \Big) +  \cE_2
	\end{split}\]
	with $\pm  \cE_2 \leq C N^{1-2\a}$. With \eqref{eq:defomegaN}
	We conclude that
	\begin{equation}\label{eq:1lineG} \begin{split} 
			& \sum_{p \in \L^*_+} \eta_p  \Big[p^2 \eta_p + N\widehat{V} (p/e^N) + \frac 1 {2} \sum_{\substack{r \in \L^*\\ p+r \in \L^*_+}} \widehat{V} (r/e^N)  \eta_{p+r}\Big]\Big(\frac{N-\cN_+}{N}\Big) \Big(\frac{N-\cN_+ -1}{N}\Big)  \\ 
			& = \frac 1 {2}  \left[ \widehat{V}(\cdot/e^N)*\widehat{f}_{N,\ell}\big)(0)- \widehat{V} (0) \right] (N-\cN_+ -1) \left( N-\cN_+ \right)  + \frac 12\sum_{p\in \L^*_+} \widehat{\o}_N(p)\eta_p  + \cE_3 \end{split} \end{equation}
	where $\pm \cE_3 \leq C N^{1-2\a}$. Using again \eqref{eq:eta-scat}, the fourth line of \eqref{eq:proofGNell-1} reads:
	\begin{equation}\begin{split} \label{eq:proofGNellQ}
			& \sum_{p \in \L^*_+} \Big[\, p^2 \eta_p + \frac N2 \widehat{V} (p/e^N) + \frac 12 \sum_{r \in \L^*:\; p+r \in \L^*_+}\hskip -0.5cm \widehat{V} (r/e^N)  \eta_{p+r} \, \Big]  \big( b^*_p b^*_{-p} + b_p b_{-p} \big) \\
			& \quad = \sum_{p \in \L^*_+} \Big[ Ne^{2N} \lambda_\ell \widehat{\chi}_\ell (p) +e^{2N} \lambda_\ell \sum_{q \in \Lambda^*} \widehat{\chi}_\ell (p-q) \eta_q - \frac 1 2 \widehat{V} (p/e^N)  \eta_{0}  \Big] \big( b^*_p b^*_{-p} + b_p b_{-p} \big)\,.
	\end{split} \end{equation}
	We focus on the last two terms on the r.h.s. of \eqref{eq:proofGNellQ}. With $\sum_{ p\in  \L_+^*}|\widehat{V}(p/e^N)|/|p|^2 \leq CN$ and $|\eta_0| \leq CN^{-2\alpha}$ (from Lemma \ref{lm:eta}), we have 
	\[ \begin{split}
		\Big| \sum_{p \in \L^*_+}&\widehat{V} (p/e^N)  \eta_{0} \big( b^*_p b^*_{-p} + b_p b_{-p} \big) \Big| \\
		& \leq  C N^{-2\a}  \bigg[  \sum_{p \in \L^*_+} \frac{|\widehat{V} (p/e^N)|^2}{p^2} \bigg]^{1/2}  \bigg[  \sum_{p \in \L^*_+} p^2 \| a_p \xi \|^2 \bigg]^{1/2}  \|(\cN_++1)^{1/2} \xi \| \\
		& \leq C N^{1/2 - 2\a} \| \cK^{1/2}\xi\|  \|(\cN_++1)^{1/2} \xi \|\,.
	\end{split}\]
	The second term on the right hand side of \eqref{eq:proofGNellQ} can be bounded in position space: 
	\[\begin{split}
		& \Big| \langle \xi,\; e^{2N} \lambda_\ell \sum_{p\in \L^*_+} (\widehat{\chi}_\ell * \eta)(p) (b_p^* b_{-p}^* + b_p b_{-p}) \xi \rangle \Big|  \\
		&\quad\leq CN^{2\a-1} \| (\cN_+ + 1)^{1/2} \xi \| \int_{\L^2} dxdy \,\chi_\ell(x-y) |\check{\eta}(x-y)|\|(\cN_+ + 1)^{-1/2}\check{b}_x\check{b}_y\xi\|\\ 
		&\quad\leq CN^{\a-1}\|(\cN_+ + 1)^{1/2} \xi \| \left[\int_{\L^2} dxdy \, \chi_\ell(x-y) \|(\cN_+ + 1)^{-1/2}\check{a}_x\check{a}_y\xi\|^2\right]^{1/2}.
	\end{split}\]
	The term in parenthesis can be bounded as (see \cite[Eq. (80)]{CCS} for details) 
	\[
	\int_{\L^2} dxdy \, \chi_\ell(x-y) \|(\cN_+ + 1)^{-1/2}\check{a}_x\check{a}_y\xi\|^2 \leq Cq N^{-2\a/q'}\|\cK^{1/2}\xi\|^2
	\]
	for any $q >2$ and $1 < q' < 2$ with $1/q +1/q' =1$. Choosing $q =\log N$, we get
	\[ \begin{split}
		\Big| \langle \xi,e^{2N} \lambda_\ell &\sum_{p\in \L^*_+} (\widehat{\chi}_\ell * \eta)(p) (b_p^* b_{-p}^* + b_p b_{-p}) \xi \rangle \Big|   \\
		&\leq C N^{-1} (\log N)^{1/2}  \| (\cN_+ + 1)^{1/2} \xi \| \|\cK^{1/2}\xi\| \,,
	\end{split}\]
	Combining the previous bounds with (\ref{eq:proofGNellQ}) and using the definition \eqref{eq:defomegaN} we obtain:
	\begin{equation} \label{eq:2lineG}
		\begin{split}
			\sum_{p \in \L^*_+} &\Big[p^2 \eta_p + \frac N2 \widehat{V} (p/e^N) + \frac 1 {2} \sum_{\substack{r \in \L^* : \\ p+r \in \L^*_+}}  \widehat{V} (r/e^N)  \eta_{p+r} \Big]  \big( b^*_p b^*_{-p} + b_p b_{-p} \big) \\ &= \frac 1 2 \sum_{p \in \L^*_+} \widehat{\o}_N (p)\big( b^*_p b^*_{-p} + b_p b_{-p} \big) + \cE_4\,,
	\end{split}\end{equation}
	with \[ | \langle \xi, \cE_4  \xi \rangle| \leq C  N^{-1} (\log N)^{1/2}  \|(\cN_++1)^{1/2}\xi\| \| \cK^{1/2}\xi \|,\] 
	if $\alpha > 1$. 
	Combining (\ref{eq:proofGNell-1}) with \eqref{eq:1lineG} and (\ref{eq:2lineG}) we conclude the proof of Prop. \ref{prop:GN}.
\end{proof}

We are now ready to complete the proof of Prop. \ref{prop:RN}; to this end, we have to control the action of the cubic conjugation $e^A$. 
\begin{proof}[Proof of Prop. \ref{prop:RN}]  The proof is based on \cite[Sec. 6]{CCS}, 
except for a few changes that we describe below. We decompose 
	\[ \label{eq:GNeff-deco}
	\cG_{N} =  \widetilde{\cO}_{N} + \cK  +\cZ_N+ \cC_{N} + \cV_N + \cE_\cG
       \]
	with $\cK$ and $\cV_N$ as in (\ref{eq:cKVN}), $\cE_\cG$ as in \eqref{eq:GeffE}, and with 
	\[\begin{split}\label{eq:GNeff-deco2}
			\widetilde{\cO}_{N}  =&\;\frac N2 \big(\widehat{V}(\cdot/e^N)*\widehat{f}_{N,\ell}\big)(0) (N-1)\left(1-\frac{\cN_+}{N}\right)  + \frac 12\sum_{p\in \L^*_+} \widehat{\o}_N(p)\eta_p\\
			& + \left[2 N\widehat{V} (0)-\frac N2 \big(\widehat{V}(\cdot/e^N)*\widehat{f}_{N,\ell}\big)(0) \right] \, \cN_+ \, \left(1-\frac{\cN_+}{N}\right) \\
			\cZ_N = &\; \frac 1 2 \sum_{ p\in  \L_+^*}\widehat{\o}_N(p)(b_pb_{-p}+\hc)\\ 
			\cC_{N} =&\;  \sqrt{N} \sum_{p,q \in \L^*_+ : p + q \not = 0} \widehat{V} (p/e^N) \left[ b_{p+q}^* a_{-p}^* a_q  + \hc \right] \, .		
	\end{split}\]
Here, we take advantage of the analysis in \cite[Sec. 6]{CCS} where properties of $e^{-A}(\cO_N+ \cZ_N+\cC_N + \cK+ \cV_N )e^A$ were established, with
\[ 
\cO_{N}  =\;\frac 1 2\widehat \o_N(0) (N-1)\left(1-\frac{\cN_+}{N}\right)   + \left[2 N\widehat{V} (0)-\frac 12 \widehat \o_N(0) \right] \, \cN_+ \, \left(1-\frac{\cN_+}{N}\right)\,.
\]
In fact, since the operators  $\cO_{N}$ and $\widetilde{\cO}_N$ only differs for some constant terms and for the fact that $\widehat \o_N(0)$ in $\cO_{N}$  is replaced by $N \big(\widehat{V}(\cdot/e^N)*\widehat{f}_{N,\ell}\big)(0) $ in $\widetilde \cO_{N}$, one can easily check that the analysis of \cite[Sec. 6]{CCS} also apply here. One conclude (see  \cite[Sec. 6.6]{CCS}):
	\[\label{eq:cReff-step1}
		\begin{split} 
		 e^{-A}&(\widetilde\cO_N+\cZ_N+\cC_N +\cK + \cV_N)e^A\\ =&\;  \frac N2 \big(\widehat{V}(\cdot/e^N)*\widehat{f}_{N,\ell}\big)(0) (N-1)\left(1-\frac{\cN_+}{N}\right)  + \frac 12\sum_{p\in \L^*_+} \widehat{\o}_N(p)\eta_p\\
		& + \frac N 2  \big(\widehat{V}(\cdot/e^N)*\widehat{f}_{N,\ell}\big)(0) \,\cN_+ \left( 1 - \frac{\cN_+}N \right) \\
			& +  \widehat \o_N(0) \sum_{p\in \L^*_+}a^*_pa_p \Big(1-\frac{\cN_+}{N} \Big)+  \frac 12 \sum_{p\in \L^*_+}  \widehat{\o}_N(p)\big[ b^*_p b^*_{-p} + b_p b_{-p} \big]  \\
			& +\frac 1 {\sqrt N} \sum_{\substack{r,v\in \L^*_+:\\ r\neq-v} } \widehat{\o}_N(r)\big[ b^*_{r+v}a^*_{-r} a_v + \text{h.c.}\big]  + \cH_N  + \cE^{(1)}_\cR\,.
		\end{split}
	\] 
with
	\be \label{eq:ReffE-1}
	\pm  \cE_\cR^{(1)} \leq C  N^{-1/2} (\log N)^{1/2} (\cH_N +1)  \, , 
	\ee
for any $\a \geq 5/2$ and $N$ sufficiently large. To conclude, we note that with Lemma \ref{lm:ANgrow},  
\[ \begin{split}
e^{-A}\cE_\cG e^A & \leq C (\log N)^{1/2} e^{-A} \big( N^{-3/2} \cH_N + N^{-1/2}(\cN_++1) \big) e^A \\
&\leq  C (\log N)^{1/2} \big( N^{-1/2} ( \cH_N +1 )  + N^{-1/2} (\cN_++1) \big)
\end{split}\]
which together with \eqref{eq:ReffE-1} leads to \eqref{eq:cReff} and \eqref{eq:ReffE}. 
\end{proof}

Finally, we  show Lemma \ref{GN-localization}, which is used in Sect. \ref{sec:proof} to localize in the number of excitations and to prove lower bounds on the spectrum of the excitation Hamiltonian. 

\begin{proof}[Proof of Lemma \ref{GN-localization}]  For simplicity, we omit the argument of the functions $f_M(\cN_+)$ and $g_M(\cN_+)$. From a direct computation, we find 
\[ \cG_N - E_N^\text{Bog} = f_M (\cG_N - E_N^\text{Bog}) f_M + g_M (\cG_N - E_N^\text{Bog}) g_M + \cE_M \]
with 
\[\cE_M= \frac{1}{2}\Big([f_M,[f_M,\cG_N]]+[g_M,[g_M,\cG_N]]\Big)\,. \]

From (\ref{eq:GNeff}), we find (with $h$ either $f$ or $g$)
\be\label{eq:commhM} \begin{split}[h_M,[h_M,\cG_N]] = &\; \frac 1 2 \sum_{ p\in  \L_+^*}\widehat{\omega}_N(p)[h_M,[h_M,(b_pb_{-p}+\hc)]]   \\
	&+ \sqrt{N} \sum_{\substack{p,q \in \L^*_+ :\\ p + q \not = 0}} \widehat{V} (p/e^N) [h_M,[h_M,(b_{p+q}^* a_{-p}^* a_q  + \hc)]]  \\
	& + [h_M,[h_M,\cE_\cG]],,
	\end{split}\ee
since all the other terms commute with $\cN_+$. For the commutator with the error term $\cE_\cG$, satisfying the bound in Eq. \eqref{eq:GeffE}, one can argue as in \cite[Prop. 4.2]{BBCS3}, and deduce that 
\[
\begin{split}
	\big| \langle \xi, [h_M,[h_M,\cE_{\cG}]]\, \xi \rangle \big| \leq\; & C M^{-2}\big( N^{1 -\a} + N^{-1}(\log N)^{1/2} \big) \|h_M'\|^2_\infty\| \cH_N^{1/2}\xi\| \| (\cN_++1)^{1/2}\xi \|\,, \\
\end{split}
\]
for any $\a >1$. The point here is that the proof of (\ref{eq:GeffE}) is based on an expansion of $\cE_\cG$ in a sum of terms given by products of creation and annihilation operators, whose commutator with $\cN_+$ has exactly the same form, up to a constant (given by the difference between the number of creation and annihilation operators in the term). Thus, each contribution to the commutator can be estimated as the corresponding term in the expansion for $\cE_\cG$ (the only difference is that terms where the number of creation operators match the number of annihilation operators do not contribute to the commutator).   

As for the quadratic off-diagonal term appearing in \eqref{eq:commhM} we have
\be\label{eq:offdiag}\Big|\langle \xi,\frac 1 2 \sum_{ p\in  \L_+^*}\widehat{\omega}_N(p)(b_pb_{-p}+\hc)\xi \rangle\Big| \leq C (\log N)^{1/2} \| (\cN_++1)^{1/2}\xi\|\|\cK^{1/2}\xi\|\,.  \ee
Since the commutators of $\cN_+$ with this term is proportional to the off-diagonal term itself, it follows that the bound  in \eqref{eq:offdiag} also holds for the commutator with $h_M$, multiplied by a factor $M^{-2}\|h_M'\|_\infty^2$. 
The same argument holds for the cubic term appearing in the first line on the r.h.s. of \eqref{eq:commhM}. Indeed, we  can bound it in position space as 
	\[\begin{split}&\Big|\langle \xi, \sqrt{N} \sum_{\substack{p,q \in \L^*_+ :\\ p + q \not = 0}} \widehat{V} (p/e^N) (b_{p+q}^* a_{-p}^* a_q  + \hc)\xi\rangle\Big|\\
	&\hspace{2cm}\leq\sqrt N \int dx dy \, e^{2N} V(e^N(x-y)) \|a_x a_y \xi \| \| a_x \xi \| \leq C \sqrt{N}\, \| \cV_N^{1/2}\xi\| \| \cN_+^{1/2}\xi \| \,.
	\end{split}\]
This implies \eqref{eq:cEM-loc}.
\end{proof}


\section{Properties of the Scattering Function} \label{App:omega}

For a potential $V$ with finite range $R_0>0$ and scattering length $\aa$, and for a fixed $R>R_0$, we establish properties of the ground state $f_R$ of the Neumann problem
\be \label{tlf2}
\Big( -\D + \frac {1}{2} V(x) \Big) f_{R}(x) =  \l_{R}\,  f_{R}(x)
\ee
on the ball $|x|\leq R$, normalized so that $f_{R}(x)=1$ for $|x|=R$.  Lemma \ref{lm:propomega}, parts (i)--(iv) follows by setting $R=e^N N^{-\a}\ell_0$ in the following lemma. 

\begin{lemma} \label{lm:appA}
	Let $V\in L^3(\bR^2)$ be non-negative, compactly supported and spherically symmetric, and denote its scattering length by $\aa$. Fix $R>0$ sufficiently large and denote by $f_{R}$ the Neumann ground state of  \eqref{tlf2}. Set  $w_{R}=1 -  f_{R}$. Then we have 
	\be \label{eq:bounds-fR}
	0\leq f_R(x)\leq 1
	\ee
	Moreover, for $R$ large enough there is a constant $C>0$ independent of $R$ such that
	\be \label{eq:eigenvalue2}
	\left | \,\l_{R} - \frac{2}{R^2 \log(R/\aa)}  \left( 1 + \frac{3}{4}\fra{1}{\log(R/\aa)} \right) \,\right | \leq  \frac{C}{R^2} \frac{1}{\log^3(R/\aa)}  \,,
	\ee
	and
	\be \label{eq:intpotf2}  
	\left| 
	\int \di x\, V(x) f_{R}(x) - \frac{4\pi}{\log(R/\aa)}\left( 1 + \fra{1}{2\log(R/\aa)} \right)  \right| \leq  \frac{C}{\log^3 (R/\aa)}   \,.
	\ee
	Finally, there exists a constant $C>0$ such that 
	\be \begin{split} \label{eq:w-bounds}
		|w_{R}(x)| &\leq   \chi(|x|\leq R_0)+ C\, \frac{\log(|x|/R) }{\log (\aa/R)} \,\chi(R_0 \leq  |x| \leq R  )  \\[0.2cm]
		|\nabla w_{R}(x)| & \leq \frac{C}{\log (R /\aa)} \frac { \chi(|x|\leq R)} {|x| + 1}
	\end{split} \ee
	for $R$ large enough.
\end{lemma}

\begin{proof} The proof of Eqs. \eqref{eq:bounds-fR},\eqref{eq:eigenvalue2} and \eqref{eq:w-bounds} can be found in \cite[App. B]{CCS}. It remains to show Eq. \eqref{eq:intpotf2}, which needs to be improved with respect to the analogous bound provided in \cite[Lemma 7]{CCS}. The starting point for its proof is the explicit expression for $f_R$, solution to the Neumann problem \eqref{tlf2}, outside the range of the potential $V$. For any $x \in \L$ s.t. $ R_0 \leq |x| \leq R$ one gets (see \cite[Eq. (192)]{CCS})	
	\be \label{eq:exp-tlf-outsideR}
	\begin{split}
		&\left| f_R (x) -1 +\frac {\e_R^2}{4}  \left(2\log( R/|x|) - 1 +\frac{x^2}{R^2} \right)  -  \frac{\e_R^4}{16} \log(R/|x|) \left(1 + \frac{ 2x^2}{R^2} \right)  \right| \\ &\hspace{9cm} \leq C \e_R^4 (\log \e_R)^2 \,,
	\end{split}
	\ee
where $\e_R^2 =\l_RR^2$. With the scattering equation \eqref{tlf2} we write
	\[
	\int \di x\, V(x)  f_{R}(x) =  2 \int_{|x| \leq R} \di x\, \Delta  f_{R}(x)  + 2 \int_{|x| \leq R} \di x  \lambda_{R}  f_{R}(x)\,.
	\]
	Passing to polar coordinates, and using that $\Delta f_R (x) = |x|^{-1} \partial_r |x| \partial_r f_R (x)$, we find that the first term vanishes. Hence 
	\[
	\int \di x\, V(x)  f_{R}(x)= 2 \lambda_R \int \di x \,  f_{R} (x) \,. \label{B.29}
	\]
	With Eq. \eqref{eq:exp-tlf-outsideR}, denoting
	\[h(r) =  2\log( R/r) - 1 +\frac{r^2}{R^2} -  \frac{\e_R^2}{2} \log(R/r) \left(\fra 12 + \frac{ r^2}{R^2} \right)\]
	and noting that 
	\[
	\int_{R_0}^R h(r) r \di r = \int_{R_0/R}^1 h(Rr) R^2r \di r =\left[\fra{r^2}{32}\left(-r^2\e_R^2-2\e_R^2+4(\e_R^2r^2+\e^2_R-8)\log r+8r^2)\right) \right]^1_{R_0/R}
	\]
	we find  
	\[\label{eq:int-lambda} \begin{split}
		\int \di x\, V(x) f_{R}(x)  &= 4\pi\lambda_R \int_{0}^{R}\di r r \Big( 1-  \frac{\e_R^2}4 h(r) + \cO( \e_R^4 |\log \e_R|^2) \Big)  \\
		& = 2\pi\lambda_R  \Big( R^2 - \fra{\e_R^2 }{8}+ \cO( \e^4_R|\log \e_R|^2 )\Big) \\
		&=  \frac{4 \pi}{ \log (R/\aa)} \Big( 1  +  \frac 3 4 \frac{1}{\log(R/\aa)} +\cO\Big(\fra1{\log^{2}(R/\aa)}\Big)\Big)\\
		&\hspace{2cm}\cdot\Big(1-\fra{1}{4}\fra{1}{\log(R/\aa)}+ \,\,\cO\Big(\frac{1}{\log^2 (R/\aa)}\Big) \Big) \,\\
		& = \frac{4 \pi}{ \log (R/\aa)} \Big(1+\fra{1}{2\log(R/\aa)}+ \,\,\cO\Big(\frac{1}{\log^2 (R/\aa)}\Big) \Big)\,,
	\end{split}\]
where in the third line we used (\ref{eq:eigenvalue2}). This concludes the proof of \eqref{eq:intpotf2}.
	
\end{proof}

\end{document}